\documentclass[12pt]{article} 
\usepackage{setspace}
\usepackage[utf8]{inputenc} 

\usepackage[margin=1in]{geometry}
\usepackage{graphicx} 
\usepackage{epstopdf}

\usepackage{booktabs} 
\usepackage{array} 
\usepackage{paralist} 
\usepackage{verbatim} 

\usepackage{fancyhdr} 
\usepackage{sectsty}
\allsectionsfont{\sffamily\mdseries\upshape} 

\usepackage[nottoc,notlof,notlot]{tocbibind} 
\usepackage[titles,subfigure]{tocloft} 

\usepackage{amssymb,amsfonts,amsmath,amsthm,tikz,mathtools}
\usepackage{circuitikz,epstopdf,subcaption,cite,enumitem}
\usepackage{pgfplots}
\pgfplotsset{compat=newest}

\newtheorem{thm}{Theorem}[section]
\newtheorem{lem}[thm]{Lemma}

\newtheorem{pro}[thm]{Proposition}
\theoremstyle{remark}

\DeclareMathOperator{\re}{Re}
\DeclareMathOperator{\im}{Im}

\DeclareMathOperator*{\argmax}{arg\,max}
\DeclareMathOperator*{\argmin}{arg\,min}

\newcommand{\mb}[1]{\ensuremath{\boldsymbol{#1}}}
\newcommand{\diag}{\mathrm{diag}}
\newcommand{\inte}{\mathrm{int}}

\usepackage{accents}
\newcommand{\ubar}[1]{\underaccent{\bar}{#1}}

\numberwithin{equation}{section}
\allowdisplaybreaks

\begin{document}
	
	\title{Solvability of Power Flow Equations Through Existence and Uniqueness of Complex Fixed Point}
	
	\author{Bai Cui
		\thanks{
			B. Cui is with the Energy Systems Division, Argonne National Laboratory, 9700 Cass Avenue, Lemont, IL 60439 (e-mail: bcui@anl.gov).
		}, Xu Andy Sun
		\thanks{X. A. Sun is with the School of Industrial and Systems Engineering, Georgia Institute of Technology, 765 Ferst Drive NW, Atlanta, Georgia 30332-0205 (e-mail: andy.sun@isye.gatech.edu).}
	}
	\vspace{-5mm}
	
	\markboth{SUBMITTED TO IEEE TRANSACTIONS ON POWER SYSTEMS, September 2014}%
	{Shell \MakeLowercase{\textit{et al.}}: Bare Demo of IEEEtran.cls for Journals}
	\maketitle
	
	\tableofcontents
	\newpage
	\begin{abstract}
		Variations of loading level and changes in system topological property may cause the operating point of an electric power systems to move gradually towards the verge of its transmission capability, which can lead to catastrophic outcomes such as voltage collapse blackout. From a modeling perspective, voltage collapse is closely related to the solvability of power flow equations. Determining conditions for existence and uniqueness of solution to power flow equations is one of the fundamental problems in power systems that has great theoretical and practical significance. In this paper, we provide strong sufficient condition certifying the existence and uniqueness of power flow solutions in a subset of state (voltage) space. The novel analytical approach heavily exploits the contractive properties of the fixed-point form in complex domain, which leads to much sharper analytical conditions than previous ones based primarily on analysis in the real domain. Extensive computational experiments are performed which validate the correctness and demonstrate the effectiveness of the proposed condition.
	\end{abstract}
	
	\section{Introduction}
	
	\subsection{Motivation}
	
	Electric power system is regarded by the National Academy of Engineering as the greatest engineering achievement in the 20th century \cite{NAE00}, which supplies electric power worldwide from generating units to end users through extremely vast and complex power networks. Maintaining the stable and reliable operation of systems with such complexity is by no means an easy task. Power systems have traditionally been designed with sufficient resilience against disturbances and contingencies. However, with ever increasing power demand and competitive electricity market, they are being operated ever closer to the operational boundaries \cite{Schavemaker08}, in other words, their loading margins to the operational boundaries are being gradually lowered. Systems with insufficient loading margins run the risk of resulting in catastrophic outcomes such as cascading failure and large-scale blackout. Several major blackouts worldwide are associated with voltage collapse --- a phenomenon manifests itself as the gradual decline of system voltage profiles followed by a sharp voltage drop that leads to system instability and collapse \cite{Dobson89}. It is known that voltage collapse is closely related to the singularity of the associated algebraic power flow equations, and the point of voltage collapse coincides with the singularity of the set of power flow equations \cite{Sauer90, Dobson94}. However, explicit characterization of the boundary of the power flow solvability set\footnote{Mathematically, for a power system modeled by quadratic power flow equations parametrized by nodal power injections, the solvability set is the set of parameters such that the quadratic system admits a `high-voltage' solution. Physically, this set describes power injections that are realizable by the networks.} is difficult: it has been shown that the solvability set can have quite complex and nonconvex structure \cite{Hiskens01}.
	
	Reliable numerical tools to calculate the distance from a given operating point to the power flow solvability boundary are available \cite{Ajjarapu92, Cutsem91}. However, the computational overhead renders them unsuitable for online applications especially under uncertain power injections where the patterns of load variations can not be precisely determined. In addition, they provide no analytical insights into how different system characteristics such as network parameters and topology, loading conditions, and generator set-points affect system steady-state stability. It remains a challenging problem to exploit the distinct properties of power flow equations and to derive strong explicit conditions under which the power flow equations admit high-voltage solutions.
	
	\subsection{Literature survey}
	
	There has been a resurgence in recent years in the search for explicit conditions certifying the existence and uniqueness of power flow solutions along the lines of works done by early pioneers in the field in the late 20th century \cite{Wu82,Chiang90,Ilic92,Chiang00}. Wu \cite{Wu82} and Ili\'c \cite{Ilic92} are among the first to derive sufficient conditions for the solvability of decoupled power flow equations in transmission system, whereas early analytical results on distribution system power flow solution existence and uniqueness have been proposed by Chiang in \cite{Chiang90,Chiang00}. Recently, energy function method and monotone operator theory has been applied to characterize convex domain in which the (non-)existence of power flow solutions can be certified \cite{Dvijotham15a,Dvijotham15b,Park19}. For decoupled real power flow equations on acyclic networks, necessary and sufficient condition for existence and uniqueness of desirable solution has recently been proposed in \cite{Dorfler13}. Sufficient solvability condition for the counterpart decoupled reactive power flow equations appears in \cite{Simpson-Porco16}. Solvability results on DC network, which shares similar model with decoupled reactive power flow model, include \cite{Barabanov16,Liu18}. Extending the analysis from decoupled power flow models to the coupled one and obtaining solvability conditions with similar quality turns out to be challenging. For coupled full power flow model, a sufficient condition for existence and uniqueness of high-voltage solution is obtained using fixed-point argument in \cite{Bolognani16}. Similar techniques have subsequently been applied to yield stronger results in \cite{Yu15,WangC16,Dvijotham18,Nguyen17}, with results in \cite{WangC16,Dvijotham18} dominating earlier ones. While the condition proposed in \cite{Dvijotham18} does not dominate the one in \cite{WangC16}, it has been shown empirically in \cite{Nguyen17} that the condition outperform the one in \cite{WangC16} most of the time. However, the improved sharpness comes at a price of no solution uniqueness guarantee. Conditions on solution existence and uniqueness in lossless radial system with voltage-controlled buses are given in \cite{Simpson17a, Simpson17b}. Extensions of the conditions to multi-phase distribution systems appear in \cite{WangC17,Bernstein18}. For a comprehensive and up-to-date summary of research on power flow solvability, see \cite{Simpson17b}. As mentioned in \cite{Simpson17b}, we now have a fairly good understanding of solution existence and uniqueness for decoupled power flow models, while the quest for sharp analytical conditions for coupled full power flow model, despite substantial research efforts \cite{Bolognani16,Yu15,WangC16,Dvijotham18,Nguyen17,Simpson17a, Simpson17b,WangC17,Bernstein18}, remains open. Apart from gaining deeper theoretical understandings of power flow solvability problem, these developed conditions are suitable for real-time monitoring and fast screening of voltage instability, as well as characterizing system stress level.
	
	\subsection{Contributions}
	
	In this work, we propose explicit sufficient solvability condition on nodal power injections that certify existence and uniqueness of solutions to power flow equations in a subset of state (voltage) space for given generator voltages and network topology. The condition relates system topology and network parameters, load power injections, and generator voltage set-points, and reveals their interplay in characterizing system stability level. For scenarios in which the existence and uniqueness of power flow solution can be certified, the condition also provides rigorous bound inside which the solution lies. The proposed condition significantly improves earlier conditions on power flow solvability. Specifically, the main contributions of the work are:
	\begin{enumerate}[label = \arabic*)]
		\item The proposed condition is shown to dominate all known solvability conditions \cite{Bolognani16,Yu15,Simpson-Porco16,WangC16,Dvijotham18,Nguyen17}. Specifically, we have analytically shown that it dominates the two strongest conditions reported in \cite{WangC16} and \cite{Dvijotham18}. In addition, unlike some existing conditions (for example, \cite{Dvijotham18}) which only guarantees power flow solution existence, the proposed condition guarantees solution existence and uniqueness within a desirable set in voltage space, characterizes a voltage subset devoid of solutions, and provides convergence guarantee for the iterative power flow algorithm.
		\item As far as we know, the proposed solvability condition is the first one to encode the effects of line resistance-to-reactance ratio and load power factors, as well as their interplay, on system solvability. As such, it serves as a better indicator on the effectiveness of different control actions for system stability and security enhancement. It can also be used as an on-line system stress monitoring tool, which provides an improved conservative estimate on the distance to steady-state feasibility and stability boundaries.
		\item Unlike previous conditions \cite{Bolognani16,Simpson-Porco16,WangC16,Simpson17b} that rely on Banach fixed point theorem for solution uniqueness, we develop a novel sufficient condition on solution uniqueness for holomorphic fixed-point equations in $\mathbb{C}^n$ that is significantly less restrictive. This general result is an extension of \cite[Thm. 6.12a]{Henrici74} from $\mathbb{C}$ to $\mathbb{C}^n$. We believe the technique is general enough to prove useful for other problems whose models display similar structural and numerical properties.
	\end{enumerate}
	
	\subsection{Applications}
	
	The condition can find a multitude of applications in power system operations and control. We briefly discuss some of them here. The interested readers can refer to \cite{Simpson-Porco16,WangC16,Nguyen17,Simpson17b} for further discussions on potential applications.
	\begin{enumerate}[label = \arabic*)]
		\item Power system contingency analysis is routinely performed by system operators to assess the system's resilience to withstand possible component (generator, transmission line, etc.) failures. To access the potential impact of possible contingencies on system steady-state response, power flow analyses need to be performed, which can be time consuming. The proposed condition can be used to certify scenarios for which the power flow solution exists and lies inside the operational constraints with minimal computational overhead, so that a large number of scenarios can be pre-screened. 
		\item To evaluate system stability and resilience against projected variations of nodal power injections, the standard computational tool is continuation power flow (CPF) \cite{Ajjarapu92}, which employs a predictor-corrector scheme that perform a sequence of power flow computations until power flow Jacobian singularity. Heuristic proxy of CPF exists, which tries to extrapolate the PV curve using a reduced number of power flow solutions \cite{Chiang97}. However, these methods are not applicable when the power injections are uncertain. On the other hand, the proposed condition provides rigorous sufficient certificates to ensure the feasibility of uncertain power injections.
		\item For the sake of security and physical limitations, it is often important to make sure that the power flow solutions not only exist, but also satisfy operational constraints such as bus voltage and line flow limits. Power flow feasibility set identifies the set of power injections such that the power flow solutions are guaranteed to exist and satisfies these constraints.  While characterizing power injections whose corresponding power flow solutions satisfy operational constraints is relatively easy, finding certificates to ensure existence of power flow solutions for the set of power injections is in fact a bottleneck for designing tractable algorithms to construct the feasibility set, where the proposed condition can be useful. For applications of recent solvability conditions on power flow feasibility set characterization, see \cite{Wang18,Wang19,Nguyen18,Lee18}.
	\end{enumerate}
	Other applications where the proposed condition can be used to quantify the system stability level and to certify power flow solution existence include: preventive and corrective control selection \cite{Simpson-Porco16,Mansour13}; allowable renewable generation certification \cite{Nguyen17}; stability-constraint optimal power flow (OPF) problem \cite{Cui17}, system stress minimization problem \cite{Todescato16}, as well as robust OPF problem \cite{Molzahn18, Louca19}.

	\section{Problem Modeling} \label{pfmodel}
	
	Since we are concerned with the long-term behavior of power system governed by balance of network flows, we adopt the algebraic model that does not incorporate electro-mechanical dynamics on the generator or load side which are relevant for short-term analysis. This modeling perspective is based upon the time-scale separation principle: the transient effects that take place on the order of seconds and the long-term effects that spans minutes to hours can be considered independently --- when evaluating long-term effects we assume that the fast transients are not excited during slow changes \cite[Sect.~5.4]{Cutsem08}.
	
	\subsection{Power system model}
	
	We consider a connected and phase-balanced power system with $n+m$ buses operating in steady-state. The underlying topology of the system can be described by an undirected connected graph $(\mathcal{N}, \mathcal{E})$, where buses are modeled as nodes $\mathcal{N}$ and lines are modeled as branches $\mathcal{E} \subseteq \mathcal{N} \times \mathcal{N}$. The buses are partitioned into two distinct types: generators ($\mathcal{N}_G$) and loads ($\mathcal{N}_L$) such that $\mathcal{N}_G \cup \mathcal{N}_L = \mathcal{N}$ and $\mathcal{N}_G \cap \mathcal{N}_L = \emptyset$. We denote the number of generators and loads as $m = |\mathcal{N}_G| \ge 1$ and $n = |\mathcal{N}_L| \ge 1$, respectively, and we assume buses $1, \ldots, m$ are generator buses and buses $m+1, \ldots, m+n$ are load buses. Every bus $i$ in the system is associated with a voltage phasor $V_i = |V_i| e^{\mathrm{i}\theta_i}$ where $|V_i|$ and $\theta_i$ are the magnitude and phase angle of the voltage.
	
	For steady-state analysis, the transmission line is generally modeled by the lumped $\pi$-equivalent model which incorporates transmission line impedance, line charging capacitors, shunt elements, and transformers \cite{Wood14}. The information is encoded in the complex admittance matrix $\mb{Y} \in \mathbb{C}^{(n+m)\times(n+m)}$ relating vector of bus voltage $\mb{V}$ and vector of bus current injection $\mb{I}$ by Ohm's Law and Kirchhoff's Law as 
	\begin{equation} \label{eq:ohm}
	\mb{I} = \mb{Y}\mb{V}.
	\end{equation}
	The admittance matrix has components $Y_{ij} = -y_{ij}$ for $(i,j) \in \mathcal{E}$ and $Y_{ii} = y_{ii} + \sum_{j=1}^{n+m} y_{ij}$, where $y_{ij}$ represents the line admittance seen from bus $i$ to $j$ while $y_{ii}$ is the shunt admittance at bus $i$. The matrices of real and imaginary parts of the admittance matrix are called conductance and susceptance matrix, respectively, and are denoted as $\mb{G}$ and $\mb{B}$ such that $\mb{Y} = \mb{G} + \mathrm{i}\mb{B}$.
	
	Generator and load buses are modeled differently in power system steady-state analysis. We model a load bus as a `PQ' bus, whose real and reactive power injections are specified and the voltage phasor is to be determined. On the other hand, since generators have voltage regulation capabilities under normal operation, generator buses are generally modeled as `PV' buses, whose real power injections and voltage magnitudes are specified and the voltage angle and reactive power injections are undetermined. However, we adopt one popular assumption in voltage stability analysis regarding generator bus modeling: we model the generator buses as `$\theta V$' buses, i.e., both the voltage magnitudes and angles are specified. For a justification of the modeling assumption, see \cite{Cui17} and references therein.
	
	The power flow equations relate bus power injections $\mb{S}$ with bus voltage through the admittance matrix. First note the vector of bus power injection can be calculated based on \eqref{eq:ohm} in the following way
	\begin{equation} \label{eq:power}
	\mb{S} = \diag(\mb{V}) \mb{I}^* = \diag(\mb{V}) \mb{Y}^* \mb{V}^*.
	\end{equation}
	By singling out the real and reactive powers and rearranging terms, we obtain the power flow equations for every load bus $i \in \mathcal{N}_L$ as functions of bus voltage magnitude and phase angles as follows
	\begin{subequations} \label{eq:powerflow}
		\begin{align}
		P_i &= \sum_{j \in \mathcal{N}} |V_i||V_j| \left( G_{ij} \cos(\theta_i - \theta_j) + B_{ij}\sin(\theta_i - \theta_j) \right), \qquad i \in \mathcal{N}_L \label{eq:pf:real}\\
		Q_i &= \sum_{j \in \mathcal{N}} |V_i||V_j| \left( G_{ij} \sin(\theta_i - \theta_j) - B_{ij}\cos(\theta_i - \theta_j) \right), \qquad i \in \mathcal{N}_L\label{eq:pf:react}
		\end{align}
	\end{subequations}
	
	\textbf{The fundamental question} addressed in this paper is the solvability of \eqref{eq:powerflow}, i.e., given load bus power injections $\mb{S}_L = \mb{P}_L + \mathrm{i}\mb{Q}_L$ and generator voltage set points $\mb{V}_G$, determine whether there exists a load voltage solution $\mb{V}_L$ that satisfies \eqref{eq:powerflow}. 
	
	To gain more analytical insights into this problem, we rewrite \eqref{eq:powerflow} in an alternative fixed point form.

	\subsection{Power flow equations in fixed point form}
	We explicitly recognize the generator and load buses in bus current, voltage vectors as well as bus admittance matrix, and rewrite \eqref{eq:ohm} as
	\begin{equation}
	\begin{bmatrix}
	\mb{I}_G \\
	-\mb{I}_L
	\end{bmatrix} = 
	\begin{bmatrix}
	\mb{Y}_{GG} & \mb{Y}_{GL} \\
	\mb{Y}_{LG} & \mb{Y}_{LL}
	\end{bmatrix}
	\begin{bmatrix}
	\mb{V}_G \\
	\mb{V}_L
	\end{bmatrix}.
	\label{admit}
	\end{equation}
	Solving for $\mb{V}_L$ in \eqref{admit} yields
	\begin{equation}
	\mb{V}_L = -\mb{Y}_{LL}^{-1}\mb{Y}_{LG}\mb{V}_G - \mb{Y}_{LL}^{-1}\mb{I}_L.
	\label{equi1}
	\end{equation}
	Denote the vector of equivalent voltage as $\mb{E} := -\mb{Y}_{LL}^{-1} \mb{Y}_{LG} \mb{V}_G$ and the impedance matrix as $\mb{Z} := \mb{Y}_{LL}^{-1}$ (the invertibility of $\mb{Y}_{LL}$ is shown in \cite{WangC16}). With the definitions, (\ref{equi1}) can be rewritten as 
	\begin{equation} \label{eq:nodekvl}
	\mb{V}_L = \mb{E} - \mb{Z}\mb{I}_L.
	\end{equation}
	Substitute $\mb{I}_L = \diag^{-1}(\mb{V}_L^*) \mb{S}_L^*$ from \eqref{eq:power} in \eqref{eq:nodekvl} and multiply both sides of \eqref{eq:nodekvl} by $\diag(\mb{E})^{-1}$, we arrive at the power flow equations in fixed point form
	\begin{equation} \label{eq:fppf}
	\mb{v}_L = \mb{1} - \hat{\mb{Z}}\diag^{-1}(\mb{v}_L^*)\mb{S}_L^*
	\end{equation}
	where the normalized load bus voltages and normalized impedance matrix are defined as
	\begin{equation}
	\mb{v}_L := \diag^{-1}(\mb{E})\mb{V}_L, \qquad \hat{\mb{Z}} := \diag^{-1}(\mb{E}) \mb{Z} \diag^{-1}(\mb{E}^*).
	\end{equation}
	Notice that $\hat{\mb{Z}}$ has the unit of $\text{watt}^{-1}$, i.e. the inverse of power. Under very mild assumption that $\mb{E}$ and $\mb{v}_L$ do not contain zero elements, which hold true for any practical power systems, the solution $\mb{v}_L$ to \eqref{eq:fppf} left-multiplied by $\diag(\mb{E})$ recovers a solution $\mb{V}_L$ to \eqref{eq:powerflow}, and vice versa. Therefore, we refer to \eqref{eq:fppf} as the power flow equations in the sequel unless otherwise stated.

	
	\section{Main Result}
	
	\subsection{A New Solvability Condition} \label{sect:main:condition}
	
	In this section, we introduce the main result of the paper: a new power flow solvability condition. The general approach to derive the condition can be roughly divided into two parts: In part one, we derive a sufficient condition on load power injection to ensure the existence of power flow solutions. The general idea has some similarity with earlier works on power flow solvability \cite{Bolognani16,Simpson-Porco16}, where we cast the power flow equations in fixed point form \eqref{eq:fppf} and derive conditions under which the fixed point mapping admits a compact convex invariant set\footnote{A set $C$ is an invariant set for $x = f(x)$ if $f(C) \subseteq C$.}. Brouwer fixed point theorem can then be used to certifies the existence of fixed point. The novelty of the proposed approach lies in the way we cast the fixed point power flow equations, which ensures the existence of invariant set for a wider range of loading conditions. Therefore solution existence can be certified for an enlarged solvability set. In part two, we prove the solution in the invariant set is unique and show the fixed point iteration converges to the power flow solution linearly as long as the initial point is inside the invariant set. Due to non-conservativeness of the certified solvability set, contraction mapping theorem is not applicable in our case. Nevertheless, we propose a novel result ensuring the uniqueness of fixed point for $n$-dimensional complex functions. The result can be used to ensure the uniqueness of fixed point of the power flow equations. Contrary to earlier works focusing on analysis in real domain, our approach heavily exploits properties of complex analysis and significantly improves the solvability condition. Furthermore, we rigorously prove that the proposed condition \emph{dominates} the existing ones. We provide the main result and some discussions in this section and defer details of the proof to the Supplementary Information.
	
	Let $\hat{z}_i^\top$ denote the transpose of vector $\hat{z}_i$ and $\hat{z}_i^\top$ is the $i$th row of the normalized impedance matrix $\hat{\mb{Z}}$ in \eqref{eq:fppf}. Before presenting the main results, we first define some system stress measures pertaining to power flow solvability. The unitless matrix $\hat{\mb{Z}} \diag(\mb{S}^*_L)$ is an important metric quantifying the stress between load buses. The row sums of $\hat{\mb{Z}} \diag(\mb{S}^*_L)$ are collected by the unitless vector $\hat{\mb{Z}} \mb{S}^*_L$, which captures the aggregated system stress on each load bus. Define $\eta_i$ in $\mathbb{C}$ to be the $i$th element of $\hat{\mb{Z}} \mb{S}^*_L$
	\begin{equation} \label{eq:eta}
	\eta_i := \hat{z}_i^\top \mb{S}^*_L =\sum_{j=1}^n \hat{z}_{ij}\mb{S}^*_{Lj}, \qquad \forall i \in \mathcal{N}_L,
	\end{equation}
	and define $\xi_i$ to be the $\ell_1$ norm of the $i$th row of $\hat{\mb{Z}} \diag(\mb{S}^*_L)$
	\begin{equation}\label{eq:xi}
	\xi_i := \|\hat{z}_i^\top \diag(\mb{S}^*_L) \|_1 =\sum_{j=1}^n |\hat{z}_{ij}\mb{S}^*_{Lj}|, \qquad \forall i \in \mathcal{N}_L.
	\end{equation} 
	Both $\eta_i$ and $\xi_i$ can be seen as stress measures for bus $i \in \mathcal{N}_L$, and appear in existing solvability literature \cite{WangC16,Dvijotham18,Nguyen17}. In addition to these two stress measures, we introduce an additional one fusing both $\eta_i$ and $\xi_i$ as
	\begin{equation}\label{eq:gamma}
	\gamma_i := 2\left(\xi_i + \re(\eta_i)\right) - \xi_i^2 - |\eta_i|^2, \qquad \forall i \in \mathcal{N}_L.
	\end{equation}
	Denote the maximum of $|\eta_i|$, $\xi_i$, and $\gamma_i$ over $i \in \mathcal{N}_L$ as $\eta, \xi$, and $\gamma$, that is
	\begin{subequations}\label{eq:xietagamma}
		\begin{align}
		\eta &:= \max_{i \in \mathcal{N}_L} |\eta_i| = \max_{i \in \mathcal{N}_L}\biggl|\sum_{j=1}^n \hat{z}_{ij}\mb{S}^*_{Lj}\biggr|, \\
		\xi &:= \max_{i \in \mathcal{N}_L} \xi_i = \max_{i \in \mathcal{N}_L}\sum_{j=1}^n |\hat{z}_{ij}\mb{S}^*_{Lj}|, \\
		\gamma &:= \max_{i \in \mathcal{N}_L} \gamma_i,
		\end{align}
	\end{subequations}
	and collect $\eta_i$, $\xi_i$, and $\gamma_i$ for all $i \in \mathcal{N}_L$ into vectors $\mb{\eta}$, $\mb{\xi}$, and $\mb{\gamma}$, respectively. With the above definitions, we present the proposed solvability condition as follows:
	
	\begin{thm} \label{thm:main}
		If $\xi, \eta, \gamma$ defined in \eqref{eq:eta}--\eqref{eq:xietagamma} satisfy the following conditions
		\begin{subequations} \label{eq:main}
			\begin{align} 
			\gamma + 2\xi\eta &< 1, \label{eq:main1}\\
			\xi - \eta &\le 1, \label{eq:main2}
			\end{align}
		\end{subequations}
		and we denote two scalars $\bar{r}$ and $\ubar{r}$ as
		\begin{equation} \label{eq:r}
		\bar{r} = \sqrt{\frac{1 - \gamma + \sqrt{(1-\gamma)^2 - 4\xi^2\eta^2}}{2\xi^2}}, \quad \ubar{r} = \sqrt{\frac{1 - \gamma - \sqrt{(1-\gamma)^2 - 4\xi^2\eta^2}}{2\xi^2}},
		\end{equation}
		then the following statements concerning solutions to the power flow equation \eqref{eq:fppf} hold:
		\begin{enumerate}
			\item There is a unique solution in the following polydisc
			\begin{align} \label{eq:v_range}
			\mathbb{D}(1-\mb{\eta}; \ubar{r}\mb{\xi}):=\big\{ \mb{v} \in \mathbb{C}^n
			\; : \; \left| v_i - (1 - \eta_i)\right| \le \ubar{r}\xi_i, \;\forall i \in \mathcal{N}_L \big\};
			\end{align} 
			
			\item There are no solutions in $\mathcal{U} \setminus D(1-\mb{\eta}; \ubar{r}\mb{\xi})$ where 
			\begin{equation}
			\mathcal{U} := \left\{ \mb{v} \in \mathbb{C}^n
			\; : \; \left| (v_i - 1)/v_i \right| < \bar{r}, \;\forall i \in \mathcal{N}_L \right\}
			\end{equation}
			
			\item The fixed point iteration \eqref{eq:fppf} converges to the unique power flow solution $\hat{\mb{v}}_L \in D(1 - \mb{\eta}; \ubar{r}\mb{\xi})$ for any starting point $\mb{v}_L^0 \in \mathcal{U}$ in such a manner that
			\begin{equation}
			|\mb{v}_L^n - \hat{\mb{v}}_L| < \bar{r} \xi(\mb{S}_L) ( 1 + \mu ) \left( \frac{2\mu}{ 1 + \mu^2} \right)^{n/2}, \quad n = 0, 1, 2, \ldots,
			\end{equation}
			for some number $0 \le \mu < 1$.
		\end{enumerate}
	\end{thm}
	
	Conditions \eqref{eq:main} may seem complicated at first, however, they have some physical interpretations. For example, a higher value of $\gamma + 2\xi\eta$ indicates a more pronounced stress level with potentially less margin to the solvability boundary. The condition suggests how load power factors and system parameters interact and collectively impact system solvability. Specifically, the condition implies that system stress level is low when load power factors are out of phase with entries of the normalized impedance matrix $\hat{\mb{Z}}$. This is consistent with the general perception that high power factor and low reactive power consumption is beneficial from a stability perspective. As the transmission lines are dominantly inductive and generator voltage angles are small, $\hat{\mb{Z}}$ is dominantly imaginary, so high power factor implies $\re( \hat{\mb{Z}}\mb{S}_L^* )$, and consequently $\gamma$, is small. The quantity $\re( \hat{\mb{Z}}\mb{S}_L^* )$ is minimized when load injections are $180$ degrees out of phase with entries of the normalized impedance matrix $\hat{\mb{Z}}$. This is to be expected, as complete out-of-phase load direction for purely imaginary $\hat{\mb{Z}}$ corresponds to pure load side reactive power support. As far as we know, this is the first condition that reflects the impact of load power factors on system solvability.
	
	Proposition \ref{thm:increase} reveals an interesting fact about condition \eqref{eq:main}: if we fix the loading direction and scale the loads up to a point where $\xi - \eta = 1$, then the value $\gamma + 2\xi\eta$ is increasing along the way, and $\gamma + 2\xi\eta \ge 1$ if $\xi - \eta = 1$. In other words, when the loads are scaled along some direction, \eqref{eq:main1} is always violated ahead of \eqref{eq:main2}. Therefore, we can focus on $\gamma + 2\xi\eta$ as a system stress level indicator in on-line monitoring without worrying about \eqref{eq:main2}. Another implication of the fact is that the solvability set is connected: we can only scale the load to the point where $\gamma + 2\xi\eta = 1$, beyond which point at least one of the two constraints in \eqref{eq:main} is violated no matter how far we go. The complete proof of Theorem \ref{thm:main} can be found in Supplementary Information.
	
	\subsection{Approximation Quality}
	
	The proposed solvability condition \eqref{eq:main} characterizes an inner approximation of the true solvability set. At the same time, there are several existing solvability conditions for AC power flow equations. Among them, \cite{WangC16} and \cite{Dvijotham18} provide best certified solvability sets. It has been demonstrated numerically in \cite{Nguyen17} that the two conditions are incomparable to each other in terms of the certified solvability sets. We present Theorem \ref{thm:compare:main} in this section, which claims \emph{dominance} of the proposed condition over those in \cite{WangC16} and \cite{Dvijotham18}, and defer the proof to the Supplementary Information. A more general result is proved therein, where we show dominance even if we allow the three solvability conditions to be built in the neighborhood of any known power flow solution.
	
	We denote the certified solvability set of the proposed condition by $\mathcal{S}_p$, that is, $\mathcal{S}_p := \{\mb{S}_L \in \mathbb{C}^n :\; \mb{S}_L \text{ satisfies } \eqref{eq:main}\}$.
		On the other hand, the certified solvability set in \cite{WangC16} is
		\begin{equation} \label{eq:wang_noload}
		\mathcal{S}_w := \{\mb{S}_L \in \mathbb{C}^n :\; 4\xi < 1 \},
		\end{equation}
		whereas the certified solvability set in \cite{Dvijotham18} is
		\begin{equation} \label{eq:dj_noload}
		\mathcal{S}_d := \{\mb{S}_L \in \mathbb{C}^n :\; \sqrt{\xi} + \sqrt{\eta} \le 1 \}.
		\end{equation}
		The theorem on the quality of the three conditions can be stated as follows:
	
	\begin{thm} \label{thm:compare:main} 
			Denote the solvability set certified by the proposed condition \eqref{eq:main} by $\mathcal{S}_p$ and let the solvability sets $\mathcal{S}_w$ and $\mathcal{S}_d$ be defined as in \eqref{eq:wang_noload} and \eqref{eq:dj_noload}, respectively. We have
			\begin{enumerate}
				\item the proposed condition dominates \eqref{eq:wang_noload} and \eqref{eq:dj_noload}, or $\mathcal{S}_w \subseteq \mathcal{S}_p$ and $\mathcal{S}_d \subseteq \bar{\mathcal{S}}_p$ hold;
				\item the proposed condition strictly dominates conditions \eqref{eq:wang_noload} and \eqref{eq:dj_noload}, or $\mathcal{S}_w \subsetneq \mathcal{S}_p$ and $\mathcal{S}_d \subsetneq \bar{\mathcal{S}}_p$, when $\{ \mb{0} \} \subsetneq \mathcal{S}_p$.
		\end{enumerate}
	\end{thm}

	\section{Computational Experiments}
	
	We present three computational experiments on our main result (Theorem \ref{thm:main}) in this section. The numerical results show that the proposed solvability condition \eqref{eq:main} significantly improves the start-of-the-art in solvability literature \cite{WangC16,Dvijotham18} --- it halves the relative errors of the estimated solvability limits and provides much tighter bounds on solution locations. Standard IEEE test systems will be used for the experiments, the data of which are available in \textsc{Matpower} package \cite{Zimmerman11}, a \textsc{Matlab}-based power system steady-state analysis tool.
	
	\subsection{Solvability limit estimation}
	
	In the first computational experiment, we test the conservativeness of the proposed condition by comparing the maximum load power certified by \eqref{eq:main} against the true solvability limit. We also compare the predictive power of our condition with two sharpest conditions known so far.
	
	When talking about certifying maximum loading level, the loading direction needs to be specified. In this experiment, we assume the loading directions are consistent with the base loadings provided in the data files. Actual maximum system loading levels (or the solvability limits) along the loading directions are obtained by Continuation Power Flow (CPF) algorithm \cite{Ajjarapu92} available in \textsc{Matpower}.
	
	Let the base loading be $\mb{S}_L$, then the respective solvability limits $\lambda_w$ and $\lambda_d$ given by \eqref{eq:wang_noload} and \eqref{eq:dj_noload} are simply the minimum scaling factors such that $\lambda_w\mb{S}_L$ violates $4\xi(\lambda_w\mb{S}_L) < 1$ (or $\lambda_d\mb{S}_L$ violates $\sqrt{\xi(\lambda_d\mb{S}_L)} + \sqrt{\eta(\lambda_d\mb{S}_L)} < 1$), and are given by
		\begin{equation}
		\lambda_w = \frac{1}{4\xi(\mb{S}_L)}, \qquad \lambda_d = \frac{1}{\sqrt{\xi(\mb{S}_L)} + \sqrt{\eta(\mb{S}_L)}}.
		\end{equation} 
		However, computing the solvability limit $\lambda_p$ for the proposed condition \eqref{eq:main} is a little trickier: $\gamma_i(\lambda\mb{S})$ is quadratic in $\lambda$, so the critical load bus index $i := \argmax_{i \in \mathcal{N}_L} \gamma_i(\lambda_p\mb{S}_L)$ may vary depending on $\lambda_p$, and can not be determined by simply examining the coefficients of the quadratic equation at base loading condition. However, as discussed at the end of Section \ref{sect:main:condition}, we know from Proposition \ref{thm:increase} that when $\xi - \eta = 1$, $\gamma_i + 2\xi\eta \ge 1$ holds for at least one $i \in \mathcal{N}_L$. This suggests the general procedure to determine the solvability limit $\lambda_p$ given base loading $\mb{S}_L$ can be divided into the following four steps: 1) determine $\xi(\mb{S}_L)$ and  $\eta(\mb{S}_L)$; 2) find the scaling factor $\kappa := 1 / \left( \xi(\mb{S}_L) - \eta(\mb{S}_L) \right)$ if $\xi(\mb{S}_L) > \eta(\mb{S}_L)$, set $\kappa = 1$ if $\xi(\mb{S}_L) = \eta(\mb{S}_L)$; 3) when $\xi(\mb{S}_L) > \eta(\mb{S}_L)$, find the index set of load buses such that $\gamma_i(\kappa\mb{S}_L) + 2\xi(\kappa\mb{S}_L)\eta(\kappa\mb{S}_L) \ge 1$, and denote the index set by $\mathcal{L}$, if $\xi(\mb{S}_L) = \eta(\mb{S}_L)$, let $\mathcal{L} = \mathcal{N}_L$; 4) For each $i \in \mathcal{L}$, solve the quadratic equation $\gamma_i(\lambda_i\kappa \mb{S}_L) + 2\xi(\lambda_i\kappa \mb{S}_L)\eta(\lambda_i\kappa \mb{S}_L) = 1$, or
		\begin{equation} \label{eq:quad_lambdai}
		\left( 2\xi(\kappa\mb{S}_L)\eta(\kappa\mb{S}_L) - \xi_i(\kappa\mb{S}_L)^2 - |\eta_i(\kappa\mb{S}_L)|^2 \right) \lambda_i^2 + 2\left( \xi_i(\kappa\mb{S}_L) + \re(\eta_i(\kappa\mb{S}_L)) \right) \lambda_i = 1
		\end{equation}
		for $\lambda_i \in (0, 1]$. Then $\lambda_p = \kappa \min_{i \in \mathcal{L}} \lambda_i$.

	Estimated solvability limits $\lambda_p$ obtained with the proposed condition are compared against 1) the estimated limits $\lambda_w$ and $\lambda_d$ by the two existing conditions \eqref{eq:wang_noload} and \eqref{eq:dj_noload}, and 2) their actual counterparts. The computation results are shown in Table \ref{tb:exp1a}. The relative errors of the three conditions calculated as $(\text{actual} - \text{bound})/\text{actual}$ are tabulated in Table \ref{tb:exp1b}. Computation results from extensive test systems show the proposed condition consistently outperforms existing ones, which numerically justify Theorem \ref{thm:compare:main}.
	
	\begin{table}[!ht]
		\centering
		\caption{Lower bounds of solvability limits obtained using the proposed condition and two existing conditions versus the true solvability limits.}
		\begin{tabular}{lccccc}
			Test case       & $\lambda_p$ (proposed) & $\lambda_d$ (\hspace{1sp}\cite{Dvijotham18}) & $\lambda_w$ (\hspace{1sp}\cite{WangC16}) & Actual value \\ \midrule
			9-bus system & $2.4425$ & $1.7534$ & $1.7512$ & $2.6577$ \\
			14-bus system & $4.3246$ & $3.5384$ & $3.5229$ & $5.3320$ \\
			24-bus system & $2.3608$ & $1.6339$ & $1.6334$ & $2.7928$ \\
			30-bus system & $5.4223$ & $4.8230$ & $4.7919$ & $6.0160$ \\
			39-bus system & $2.1174$ & $1.3869$ & $1.3600$ & $2.4730$ \\
			57-bus system & $1.3456$ & $1.0998$ & $1.0935$ & $1.9074$ \\
			118-bus system & $4.7597$ & $3.9192$ & $3.9186$ & $5.4479$ \\
			300-bus system & $0.7712$ & $0.5251$ & $0.3641$ & $1.6585$ \\ 
			1354-bus system & $1.2751$ & $0.7376$ & $0.7273$ & $1.5332$ \\ 
			2383-bus system & $1.4594$ & $1.0489$ & $1.0474$ & $1.9739$ \\   \bottomrule
		\end{tabular}
		\label{tb:exp1a}
	\end{table}
	
	\begin{table}[!ht]
		\centering
		\caption{Relative errors of solvability limit approximations obtained using the proposed condition and two existing conditions}
		\begin{tabular}{lccc}
			Test case       & $\lambda_p$ (proposed) & $\lambda_d$ (\hspace{1sp}\cite{Dvijotham18}) & $\lambda_w$ (\hspace{1sp}\cite{WangC16}) \\ \midrule
			9-bus system & $8.10\%$ & $34.02\%$ & $34.11\%$ \\ 
			14-bus system & $18.89\%$ & $33.64\%$ & $33.93\%$ \\ 
			24-bus system & $15.47\%$ & $41.50\%$ & $41.51\%$ \\ 
			30-bus system & $9.87\%$ & $19.83\%$ & $20.35\%$ \\ 
			39-bus system & $14.38\%$ & $43.92\%$ & $45.01\%$ \\ 
			57-bus system & $29.45\%$ & $42.34\%$ & $42.67\%$ \\ 
			118-bus system & $12.63\%$ & $28.06\%$ & $28.07\%$ \\ 
			300-bus system & $53.50\%$ & $68.34\%$ & $78.05\%$ \\ 
			1354-bus system & $16.83\%$ & $51.89\%$ & $52.56\%$ \\ 
			2383-bus system & $26.06\%$ & $46.86\%$ & $46.94\%$ \\ 
			\textbf{Average} & $\mb{20.52\%}$ & $\mb{41.04\%}$ & $\mb{42.32\%}$ \\   \bottomrule
		\end{tabular}
		\label{tb:exp1b}
	\end{table}
	
	The computation results show that the improvement of solvability limit estimation is significant. As seen from the last row of Table \ref{tb:exp1b}, the average relative error by the proposed condition is less than half of that given by both existing methods. For most test systems, the proposed condition more than halves the relative errors. Except for 300-bus system, the relative errors for all other test systems fluctuate between $7\%$ and $30\%$. The results obtained for the 10 test systems also suggest that the relative errors are insensitive to system size. The proposed condition certifies power flow solvability under base loading condition for all systems except for 300-bus system (since the scaling factors are all greater than 1 except for 300-bus system in Table \ref{tb:exp1a}). We discuss in the Supplementary Information how to improve the solvability limit estimation using some known power flow solutions. This is particularly relevant when we are interested in certifying power flow solvability for power injections that vary around some known nominal point.
	
	\begin{figure*}[!ht]
		\centering
		\begin{subfigure}[b]{.32\linewidth}
			\centering
			\includegraphics[width=.99\textwidth]{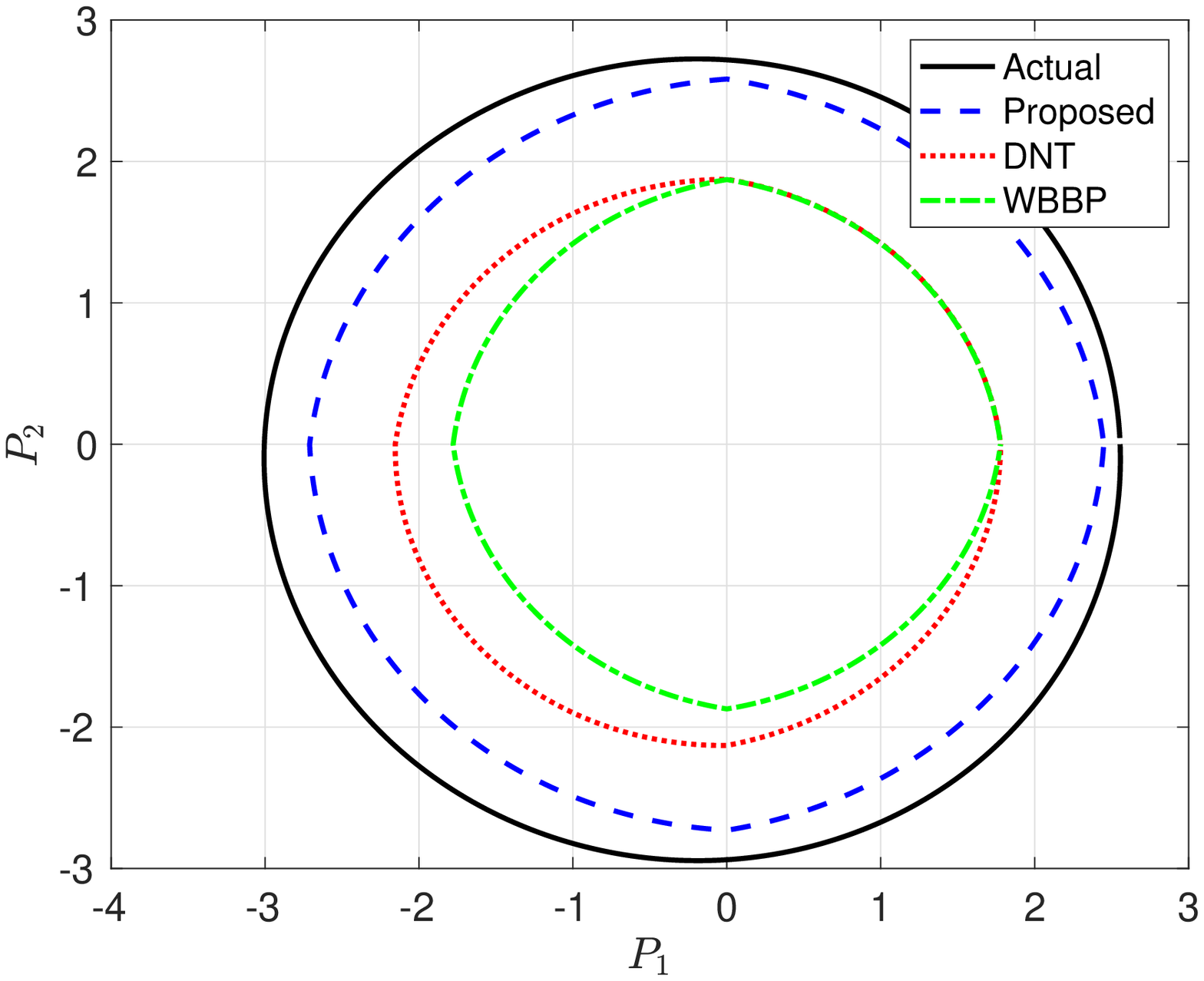}
			\caption{9-bus system}\label{fig:1a}
		\end{subfigure}%
		\begin{subfigure}[b]{.32\linewidth}
			\centering
			\includegraphics[width=.99\textwidth]{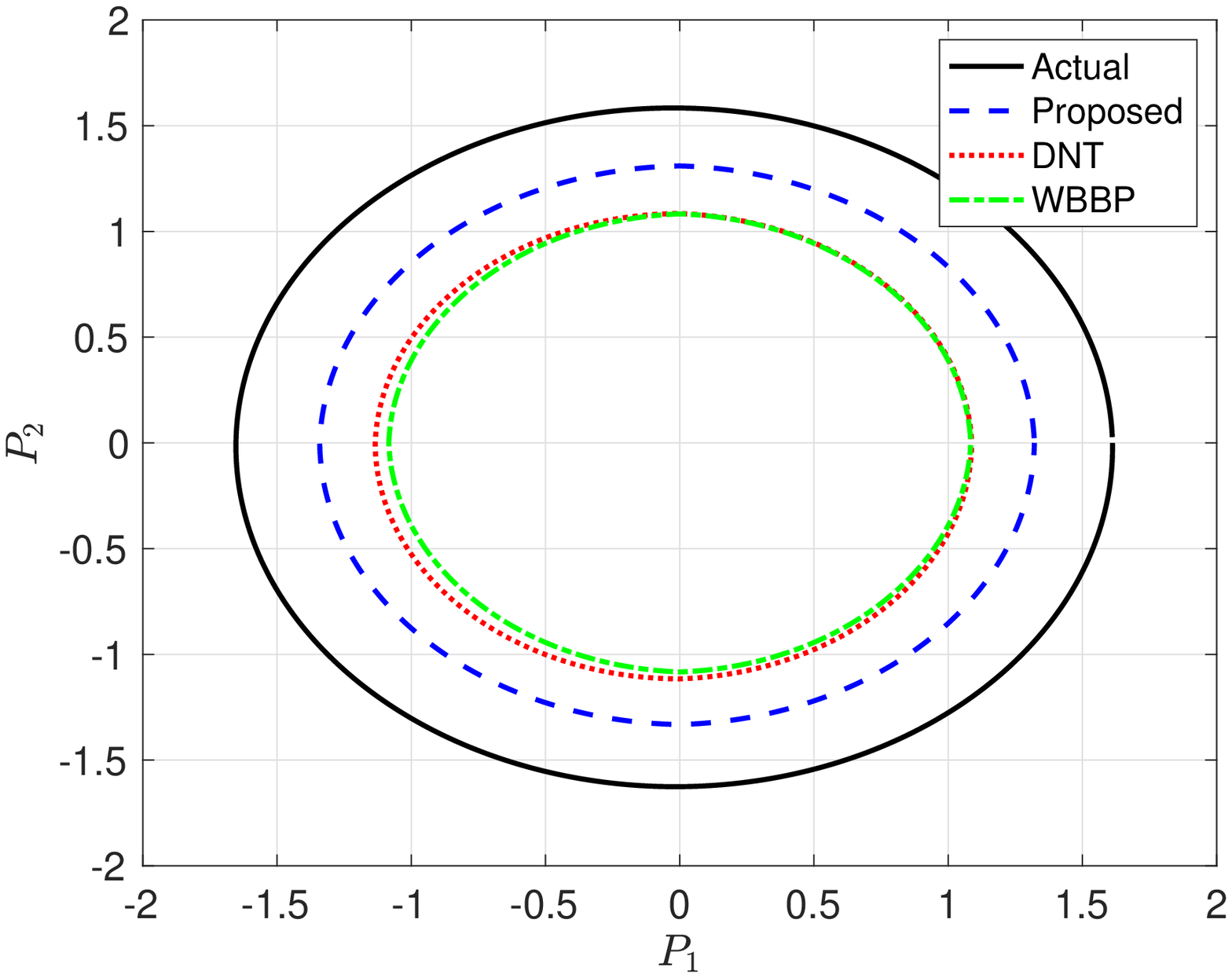}
			\caption{14-bus system}\label{fig:1b}
		\end{subfigure}%
		\begin{subfigure}[b]{.32\linewidth}
			\centering
			\includegraphics[width=.99\textwidth]{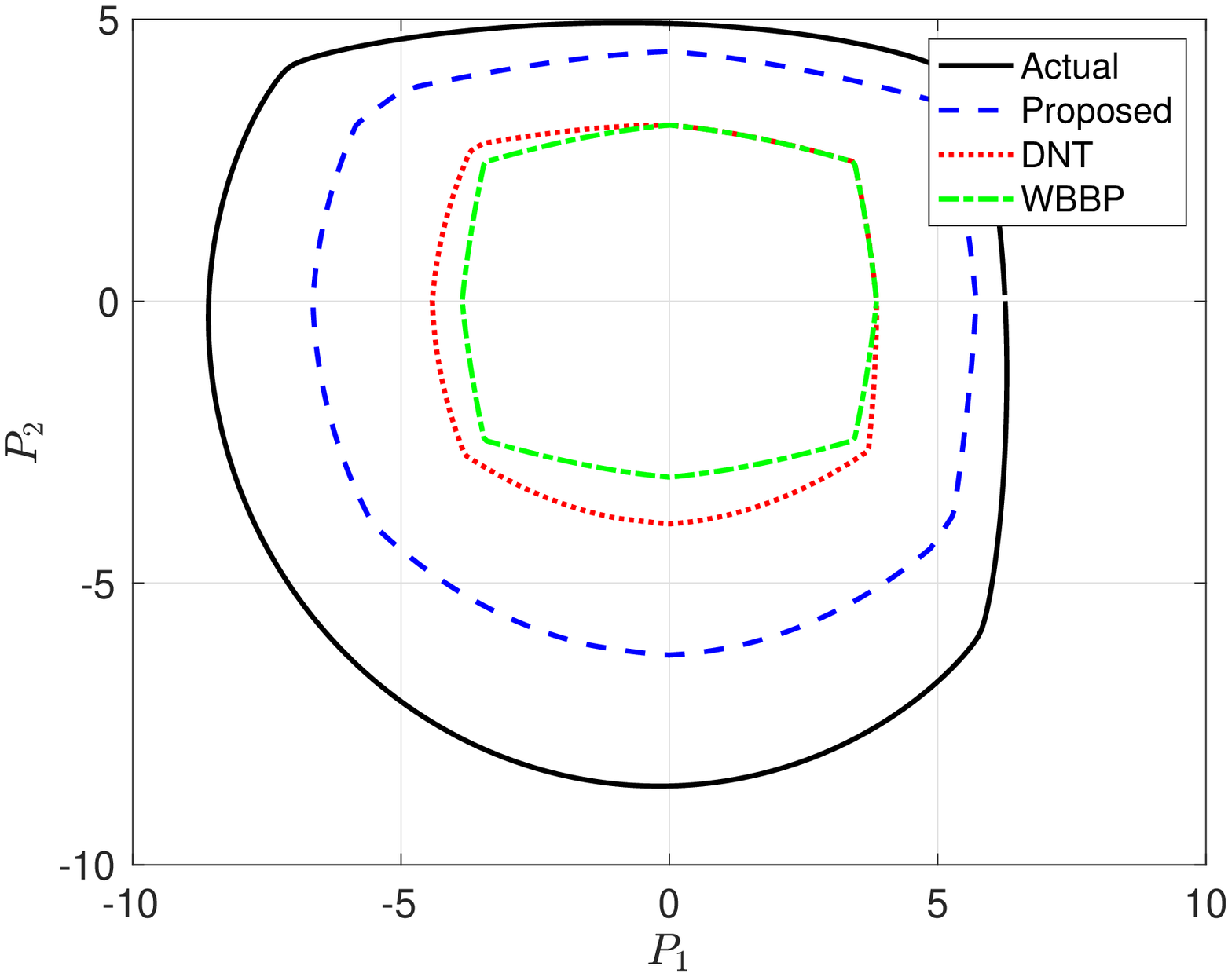}
			\caption{24-bus system}\label{fig:1c}
		\end{subfigure}\\%
		\begin{subfigure}[b]{.32\linewidth}
			\centering
			\includegraphics[width=.99\textwidth]{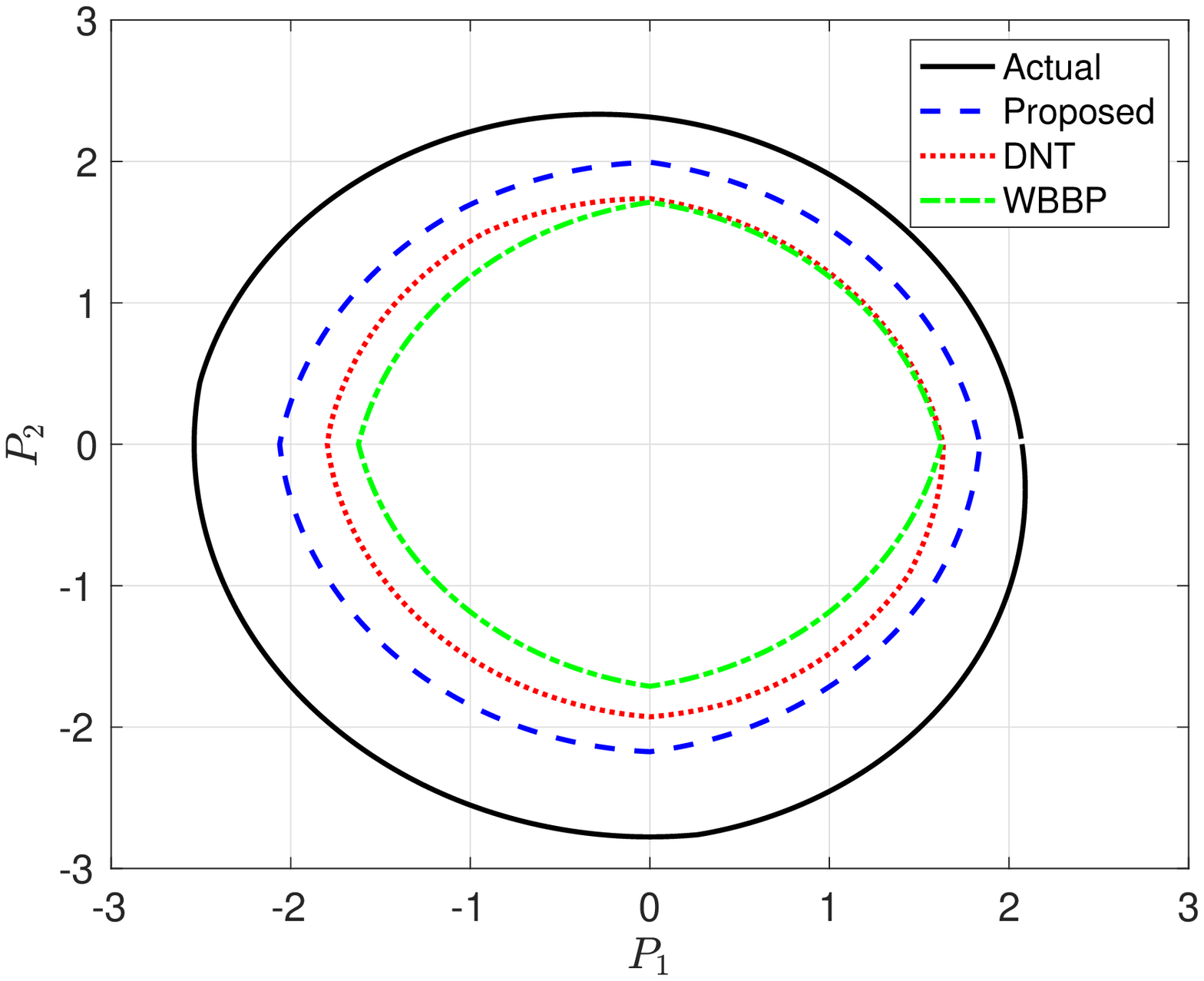}
			\caption{30-bus system}\label{fig:1d}
		\end{subfigure}%
		\begin{subfigure}[b]{.32\linewidth}
			\centering
			\includegraphics[width=.99\textwidth]{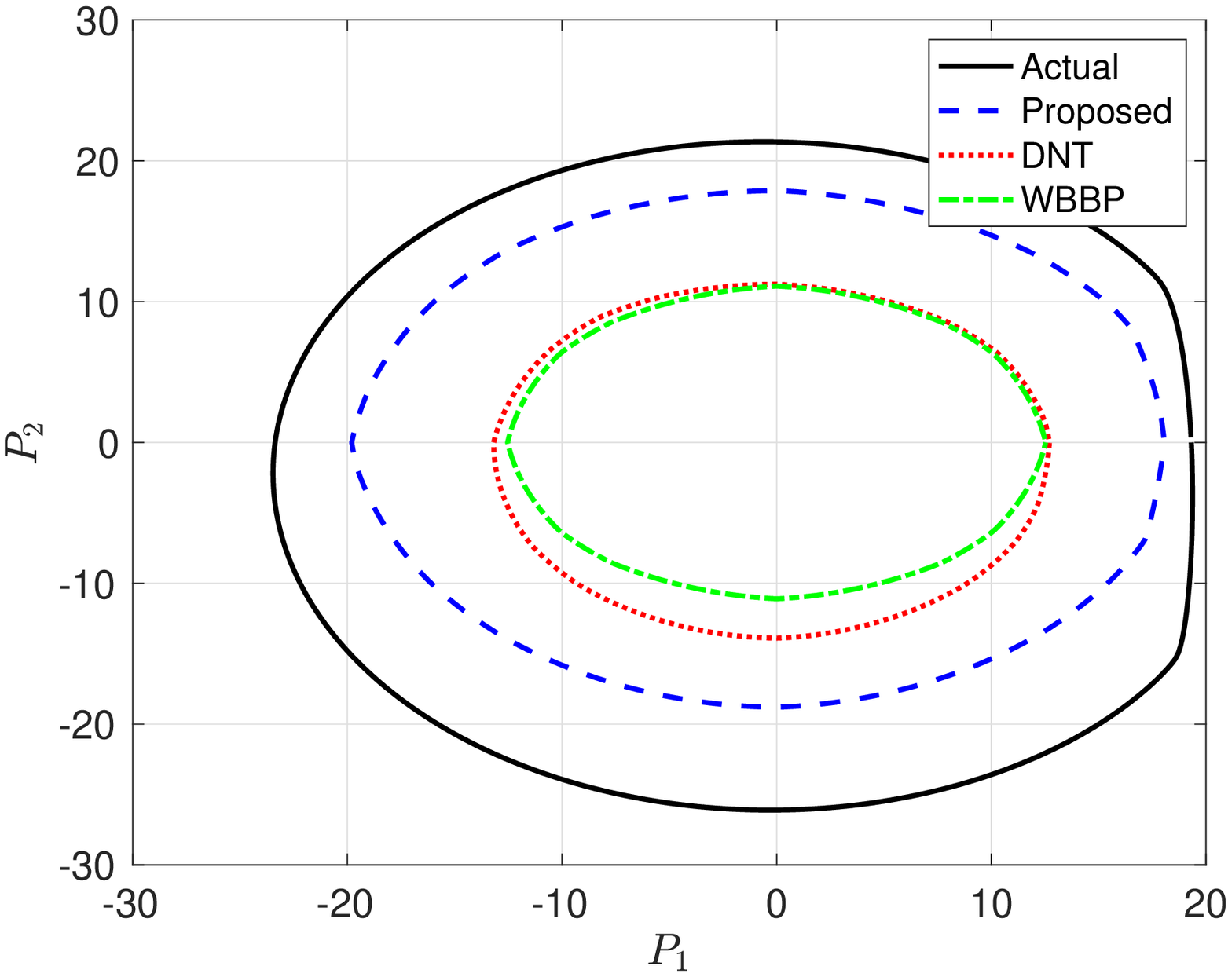}
			\caption{39-bus system}\label{fig:1e}
		\end{subfigure}%
		\begin{subfigure}[b]{.32\linewidth}
			\centering
			\includegraphics[width=.99\textwidth]{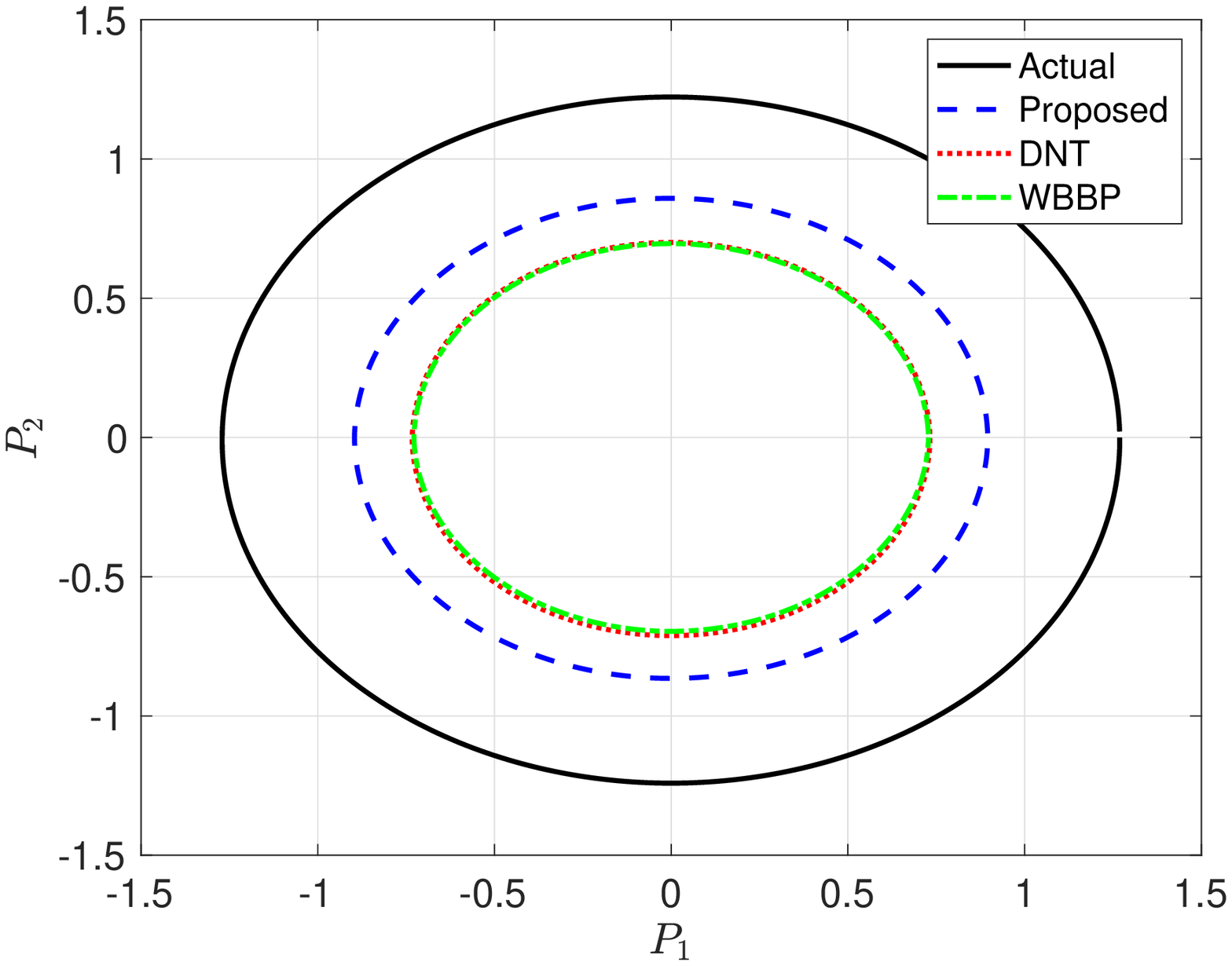}
			\caption{57-bus system}\label{fig:1f}
		\end{subfigure}\\%
		\begin{subfigure}[b]{.32\linewidth}
			\centering
			\includegraphics[width=.99\textwidth]{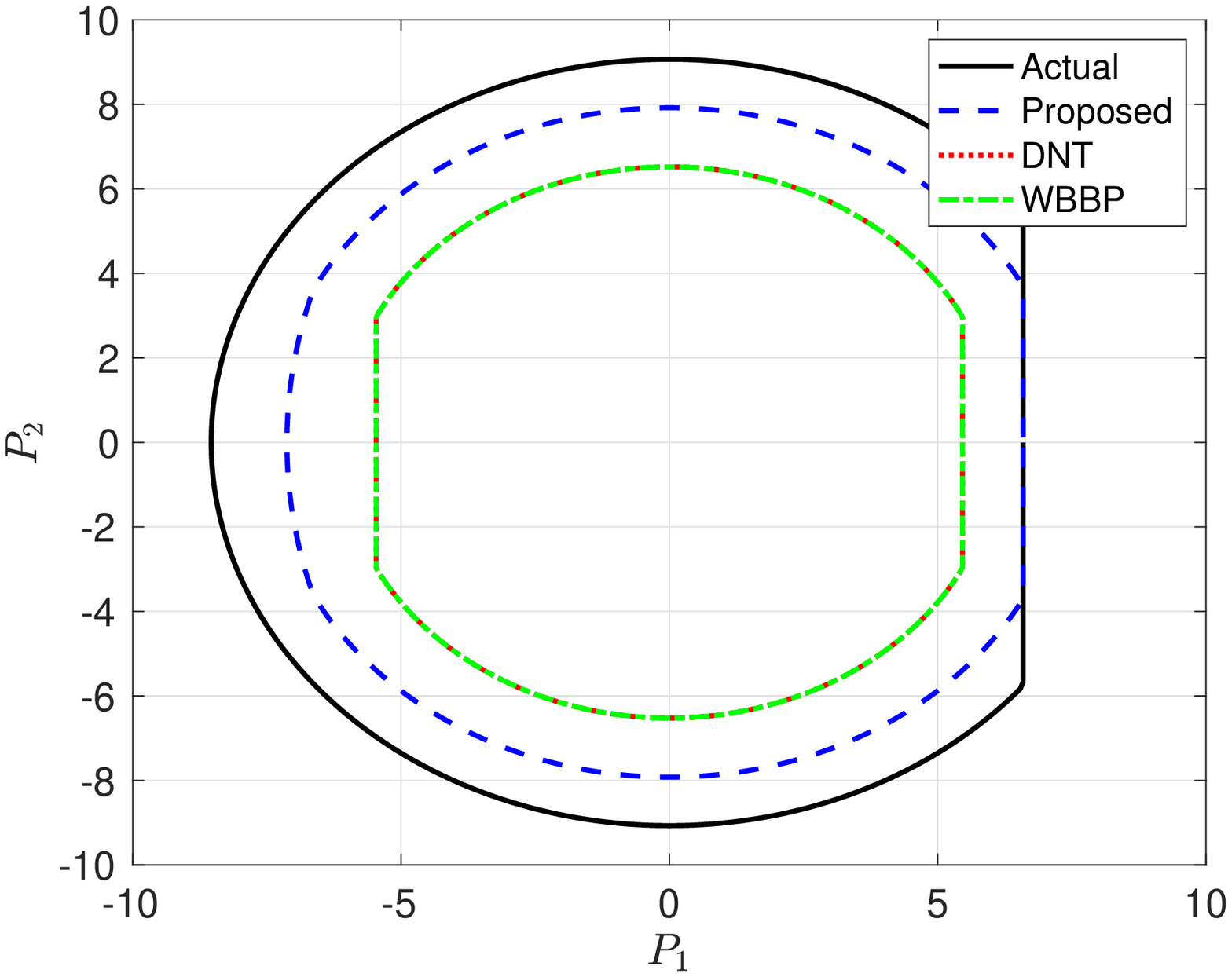}
			\caption{118-bus system}\label{fig:1g}
		\end{subfigure}%
		\begin{subfigure}[b]{.32\linewidth}
			\centering
			\includegraphics[width=.99\textwidth]{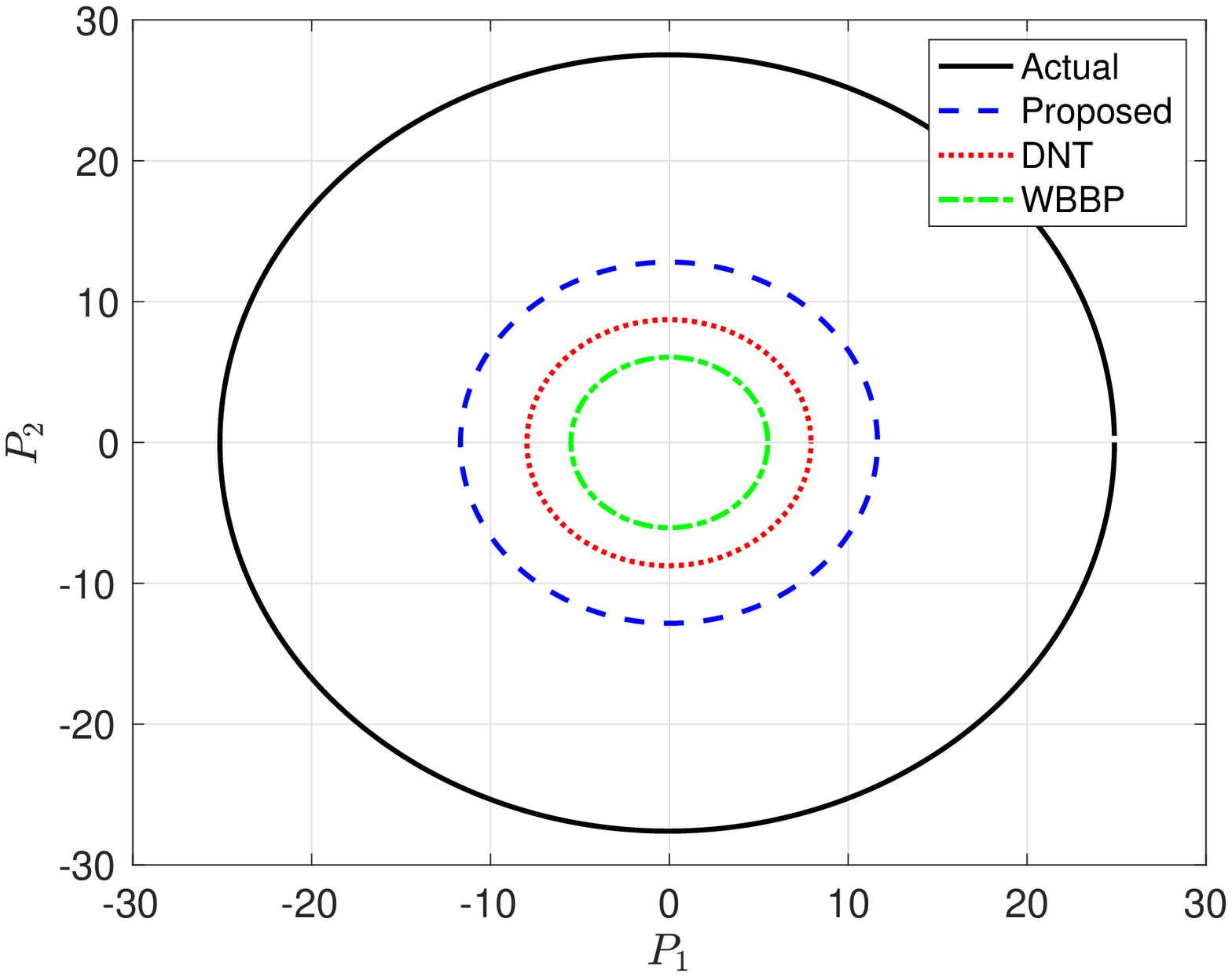}
			\caption{300-bus system}\label{fig:1h}
		\end{subfigure}%
		\begin{subfigure}[b]{.32\linewidth}
			\centering
			\includegraphics[width=.99\textwidth]{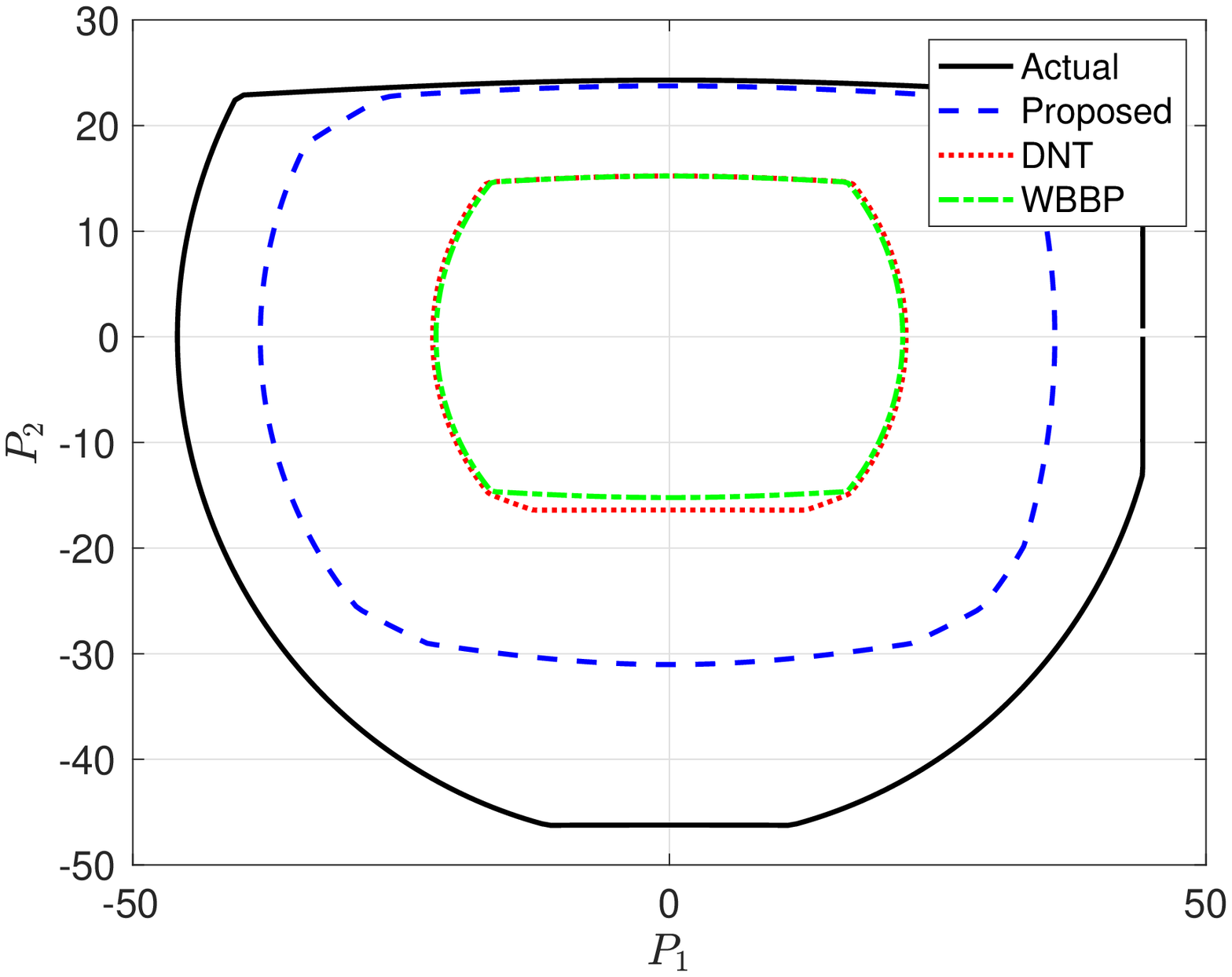}
			\caption{1354-bus system}\label{fig:1i}
		\end{subfigure}\\%
		\begin{subfigure}[b]{.32\linewidth}
			\centering
			\includegraphics[width=.99\textwidth]{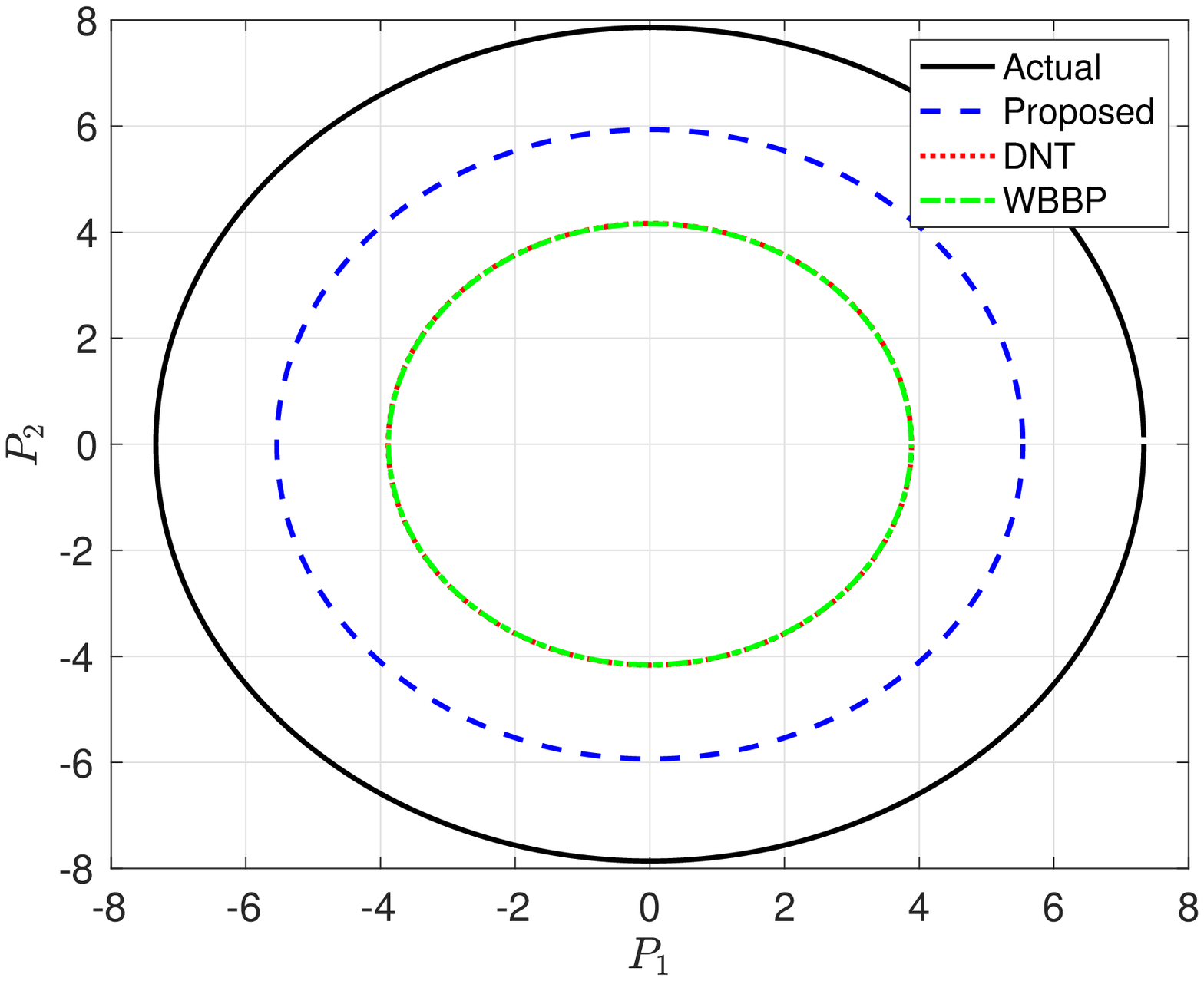}
			\caption{2383-bus system}\label{fig:1j}
		\end{subfigure}%
		\caption{Contour of solvability limit estimation for IEEE test systems.}\label{fig:contour}
	\end{figure*}
	
	\subsection{Contour of solvability limit estimation}
	
	In the second computational experiment, we perform solvability limit estimates along different loading directions and investigate the strength of the proposed condition under different loading patterns. For each test system, we change the loading directions of the first two load buses with nonzero real power demand while fixed the rest, and calculate the estimated solvability limit. To make sure the changes are pronounced enough, we artificially scale the powers of the first two loads such that they have equal magnitudes and the 2-norm of their load powers is equal to that of the rest of the load buses. By varying loading directions of the first two loads while keeping the 2-norm of their powers constant, we obtain the solvability contour as shown in Figure \ref{fig:contour}, which are the projections of the full-dimonsional solvability region to the two dimensions corresponding to the first two load buses. Similar to the first experiment, we again conclude from the simulation results that the proposed method produces the best solvability limit along all directions for all test systems. While two existing methods produce similar estimates, the proposed condition improves theirs by a wide margin. 
	
	\subsection{Voltage bound estimation}
	
	In the third computational experiment, we examine the conservativeness of the voltage bound estimation provided in Theorem \ref{thm:main}. We pick the IEEE 39-bus system for this experiment, which is a classic test system based on a reduced order New England power system commonly used for voltage stability analysis. We examine tightness of the voltage bounds under normal (base) loading conditions for buses across the system, as well as for buses in stressed system condition under progressive load increase. To this end, we develop one experiment for each scenario. 
	
	In the first experiment, we compute voltage bounds for all load buses at base loading condition based on the voltage bound \eqref{eq:v_range} in Theorem \ref{thm:main}, and compare the bounds with actual load bus voltages. We know from \eqref{eq:v_range} that the voltage upper and lower bounds for bus $i \in \mathcal{N}_L$ are given by $\bar{V}_{Li} = |E_i| (1-\eta_i + \ubar{r}\xi_i)$ and $\ubar{V}_{Li} = |E_i| (1 - \eta_i - \ubar{r}\xi_i)$, respectively, where $\ubar{r}$ is defined in \eqref{eq:r}, and the voltage angle bounds can be computed analogously. In addition, we take the center of the polydisc as the approximate voltage values. The simulation results on voltage magnitudes and phase angles are shown in Figures \ref{fig:exp3a} and \ref{fig:exp3b}, respectively. True values are shown in \ref{pgfplots:true} and approximate values are shown in \ref{pgfplots:approximate}, whereas voltage magnitude / angle bounds are marked by \ref{pgfplots:bound}. The results suggest that the error bounds of the linear power flow approximation works quite well under base loading condition, with error bounds for voltage magnitude less than $0.1$ p.u. and voltage angle less than $5$ degrees across the entire system. One more thing to note is that the voltage angle approximation is extremely accurate, the errors of which are all within $1$ degree.
	
	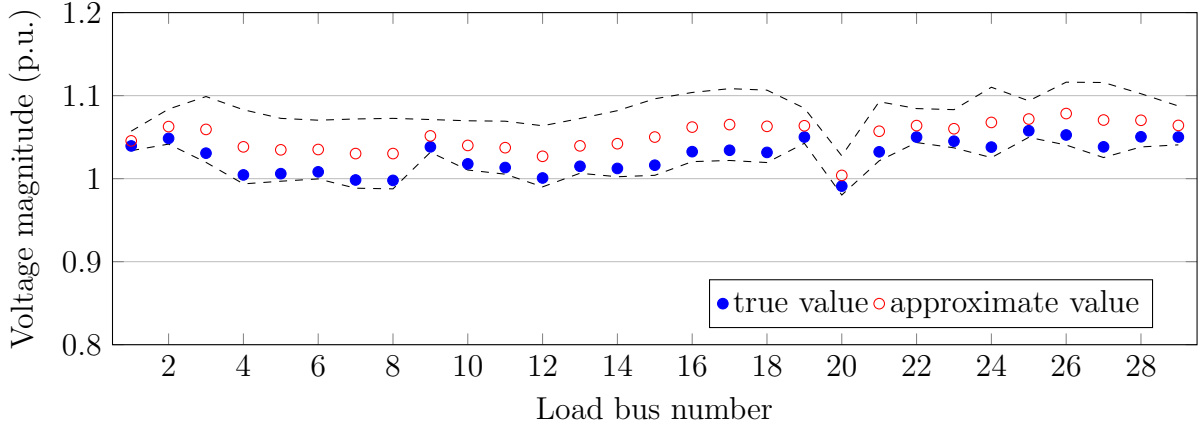
\begin{figure*}[!ht]
		\centering
		\begin{tikzpicture}
		\pgfplotsset{every axis legend/.append style={
				at={(0.55,0.05)},
				anchor=south west}}
		\begin{axis}[
		xlabel=Load bus number,
		ylabel=Voltage magnitude (p.u.),
		width=16cm,
		height=6cm,
		ymajorgrids=true,
		xmin = 0.5,
		xmax = 29.5,
		ymin = 0.8,
		ymax = 1.2,
		legend columns=4
		]
		\addplot[color=blue,mark=*, only marks] coordinates {
			(1, 1.0394) (2, 1.0485) (3, 1.0307) (4, 1.0045) (5, 1.0060) (6, 1.0082) (7, 0.9984) (8, 0.9979) (9, 1.0383) (10, 1.0178) (11, 1.0134) (12, 1.0008) (13, 1.0149) (14, 1.0123) (15, 1.0162) (16, 1.0325) (17, 1.0342) (18, 1.0316) (19, 1.0501) (20, 0.9910) (21, 1.0323) (22, 1.0501) (23, 1.0451) (24, 1.0380) (25, 1.0577) (26, 1.0526) (27, 1.0383) (28, 1.0504) (29, 1.0501) 
		};
		\addlegendentry{true value}
		\label{pgfplots:true}
		\addplot[color=red,mark=o, only marks] coordinates {
			(1, 1.0455) (2, 1.0627) (3, 1.0593) (4, 1.0383) (5, 1.0347) (6, 1.0351) (7, 1.0302) (8, 1.0302) (9, 1.0515) (10, 1.0400) (11, 1.0373) (12, 1.0269) (13, 1.0395) (14, 1.0421) (15, 1.0501) (16, 1.0621) (17, 1.0651) (18, 1.0630) (19, 1.0637) (20, 1.0039) (21, 1.0571) (22, 1.0639) (23, 1.0601) (24, 1.0677) (25, 1.0718) (26, 1.0784) (27, 1.0706) (28, 1.0703) (29, 1.0641)
		};
		\addlegendentry{approximate value}
		\label{pgfplots:approximate}
		\addplot[color = black, dashed] coordinates {
			(1, 1.0572) (2, 1.0838) (3, 1.0989) (4, 1.0830) (5, 1.0726) (6, 1.0704) (7, 1.0719) (8, 1.0726) (9, 1.0712) (10, 1.0697) (11, 1.0692) (12, 1.0637) (13, 1.0724) (14, 1.0819) (15, 1.0962) (16, 1.1038) (17, 1.1084) (18, 1.1066) (19, 1.0849) (20, 1.0275) (21, 1.0927) (22, 1.0844) (23, 1.0831) (24, 1.1101) (25, 1.0937) (26, 1.1162) (27, 1.1158) (28, 1.1023) (29, 1.0875) 
		};
		\label{pgfplots:bound}
		\addplot[color=black, dashed] coordinates {
			(1, 1.0337) (2, 1.0416) (3, 1.0196) (4, 0.9937) (5, 0.9969) (6, 0.9997) (7, 0.9885) (8, 0.9877) (9, 1.0318) (10, 1.0104) (11, 1.0054) (12, 0.9901) (13, 1.0065) (14, 1.0023) (15, 1.0040) (16, 1.0205) (17, 1.0219) (18, 1.0194)
			(19, 1.0426) (20, 0.9803) (21, 1.0215) (22, 1.0435) (23, 1.0371) (24, 1.0253) (25, 1.0498) (26, 1.0405) (27, 1.0254) (28, 1.0382) (29, 1.0407) 
		};
		\end{axis}
		\end{tikzpicture}
		\caption{Bus voltage magnitudes for IEEE 39-bus system at base loading condition. True voltage magnitudes are shown in \ref{pgfplots:true} and approximate values are shown in \ref{pgfplots:approximate}. Voltage magnitude upper and lower bounds as given in Theorem \ref{thm:main} are shown in \ref{pgfplots:bound}.} \label{fig:exp3a}
	\end{figure*}
	
	\begin{figure*}[!ht]
		\centering
		\begin{tikzpicture}
		\pgfplotsset{every axis legend/.append style={
				at={(0.55,0.05)},
				anchor=south west}}
		\begin{axis}[
		xlabel=Load bus number,
		ylabel=Voltage angle (deg),
		width=16cm,
		height=6cm,
		ymajorgrids=true,
		xmin = 0.5,
		xmax = 29.5,
		ymin = -20,
		ymax = 0,
		legend columns=4
		]
		\addplot[color=blue,mark=*, only marks] coordinates {
			(1, -13.5366) (2, -9.7853) (3, -12.2764) (4, -12.6267) (5, -11.1923) (6, -10.4083) (7, -12.7556) (8, -13.3358) (9, -14.1784) (10, -8.1709) (11, -8.9370) (12, -8.9988) (13, -8.9299) (14, -10.7153) (15, -11.3454) (16, -10.0333) (17, -11.1164) (18, -11.9862) (19, -5.4101) (20, -6.8212) (21, -7.6287) (22, -3.1831) (23, -3.3813) (24, -9.9138) (25, -8.3692) (26, -9.4388) (27, -11.3622) (28, -5.9284) (29, -3.1699) 
		};
		\addlegendentry{true value}
		\addplot[color=red,mark=o, only marks] coordinates {
			(1, -13.4652) (2, -9.6505) (3, -11.9854) (4, -12.3663) (5, -11.0173) (6, -10.2625) (7, -12.5247) (8, -13.0844) (9, -14.0184) (10, -8.0921) (11, -8.8457) (12, -9.0202) (13, -8.8343) (14, -10.5376) (15, -11.0842) (16, -9.7581) (17, -10.8215) (18, -11.6700) (19, -5.3311) (20, -6.7243) (21, -7.4667) (22, -3.1239) (23, -3.3072) (24, -9.5766) (25, -8.2493) (26, -9.2270) (27, -11.0702) (28, -5.7865) (29, -3.0906) 
		};
		\addlegendentry{approximate value}
		\addplot[color = black, dashed] coordinates {
			(1, -12.8224) (2, -8.5149) (3, -9.8403) (4, -9.9011) (5, -8.9224) (6, -8.3055) (7, -10.2028) (8, -10.7222) (9, -12.9449) (10, -6.4584) (11, -7.0824) (12, -6.9685) (13, -7.0183) (14, -8.3488) (15, -8.5686) (16, -7.5110) (17, -8.4963) (18, -9.3187) (19, -4.1917) (20, -5.3774) (21, -5.5370) (22, -2.0247) (23, -2.0637) (24, -7.2990) (25, -7.0741) (26, -7.2171) (27, -8.6489) (28, -4.0695) (29, -1.8322) 
		};
		\addplot[color=black, dashed] coordinates {
			(1, -14.1080) (2, -10.7861) (3, -14.1305) (4, -14.8314) (5, -13.1123) (6, -12.2195) (7, -14.8467) (8, -15.4466) (9, -15.0919) (10, -9.7257) (11, -10.6090) (12, -11.0719) (13, -10.6503) (14, -12.7265) (15, -13.5998) (16, -12.0053) (17, -13.1467) (18, -14.0214) (19, -6.4706) (20, -8.0712) (21, -9.3963) (22, -4.2232) (23, -4.5508) (24, -11.8542) (25, -9.4246) (26, -11.2369) (27, -13.4914) (28, -7.5034) (29, -4.3489) 
		};
		\end{axis}
		\end{tikzpicture}
		\caption{Bus voltage angles for IEEE 39-bus system at base loading condition. True voltage angles are shown in \ref{pgfplots:true} and approximate values are shown in \ref{pgfplots:approximate}. Voltage angle upper and lower bounds as given in Theorem \ref{thm:main} are shown in \ref{pgfplots:bound}.} \label{fig:exp3b}
	\end{figure*}
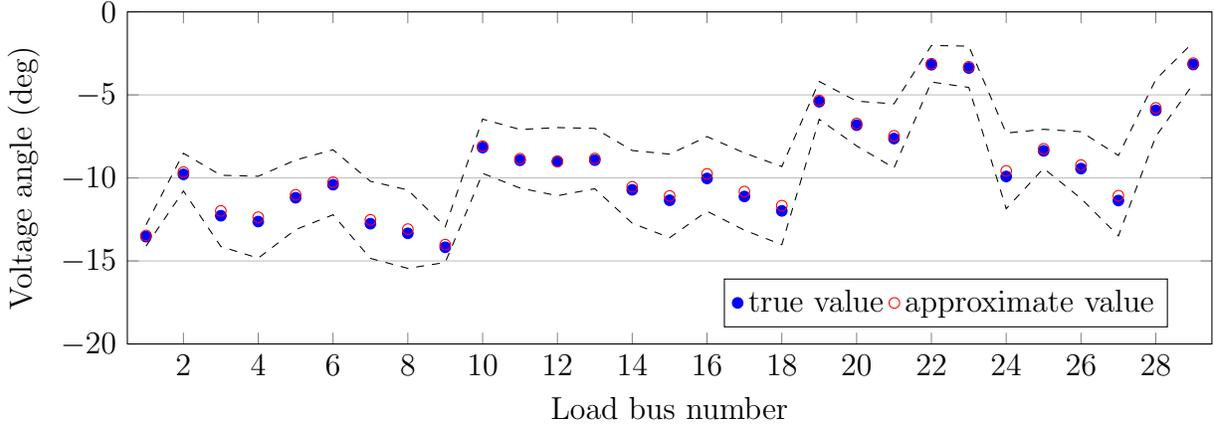
	
	We then test the proposed voltage bound estimation as system load powers progressively build up. The bound is calculated for bus 4, which is the most critical bus evaluated by the proposed condition\footnote{i.e., the load bus with index $i = \argmin_{i \in \mathcal{L}} \lambda_i$ where $\lambda_i$ is as defined in \eqref{eq:quad_lambdai}.}. Existing voltage bound estimation is available in \cite{WangC16}, and we compare their result with ours. The computational result is shown pictorially in Figure \ref{fig:exp3c}. For clarity, only voltage lower bounds are shown. As we have demonstrated in Tables \ref{tb:exp1a} and \ref{tb:exp1b}, the proposed condition provides a sharper estimate of the solvability limit, therefore able to provide voltage bounds for a larger interval of loading factors. In addition, it is observed that the quality of the estimates, both the proposed one and the one in \cite{WangC16}, degrade with increased load. However, the quality of the proposed one remains significantly better than that in \cite{WangC16} throughout system evolution, with the largest error in the order of $0.15$ p.u. at the very last voltage point.
	
	\begin{figure*}[!ht]
		\centering
		\begin{tikzpicture}
		\pgfplotsset{every axis legend/.append style={
				at={(0.5,0.02)},
				anchor=south west}}
		\begin{axis}[
		xlabel=Loading factor,
		ylabel=Voltage magnitude (p.u.),
		width=10cm,
		height=8cm,
		ymajorgrids=true,
		xmin = 1,
		xmax = 2.5,
		ymin = 0.4,
		ymax = 1.1
		]
		\addplot[color=black, very thick] coordinates {
			(1.00, 1.0045) (1.01, 1.0035) (1.02, 1.0025) (1.03, 1.0015) (1.04, 1.0004) (1.05, 0.9994) 
			(1.06, 0.9984) (1.07, 0.9974) (1.08, 0.9963) (1.09, 0.9953) (1.10, 0.9942) (1.11, 0.9932) 
			(1.12, 0.9921) (1.13, 0.9911) (1.14, 0.9900) (1.15, 0.9889) (1.16, 0.9878) (1.17, 0.9867) 
			(1.18, 0.9856) (1.19, 0.9845) (1.20, 0.9834) (1.21, 0.9823) (1.22, 0.9811) (1.23, 0.9800) 
			(1.24, 0.9788) (1.25, 0.9777) (1.26, 0.9765) (1.27, 0.9754) (1.28, 0.9742) (1.29, 0.9730) 
			(1.30, 0.9718) (1.31, 0.9706) (1.32, 0.9694) (1.33, 0.9682) (1.34, 0.9670) (1.35, 0.9657) 
			(1.36, 0.9645) (1.37, 0.9632) (1.38, 0.9620) (1.39, 0.9607) (1.40, 0.9594) (1.41, 0.9582) 
			(1.42, 0.9569) (1.43, 0.9556) (1.44, 0.9543) (1.45, 0.9529) (1.46, 0.9516) (1.47, 0.9503) 
			(1.48, 0.9489) (1.49, 0.9475) (1.50, 0.9462) (1.51, 0.9448) (1.52, 0.9434) (1.53, 0.9420) 
			(1.54, 0.9406) (1.55, 0.9392) (1.56, 0.9377) (1.57, 0.9363) (1.58, 0.9348) (1.59, 0.9334) 
			(1.60, 0.9319) (1.61, 0.9304) (1.62, 0.9289) (1.63, 0.9273) (1.64, 0.9258) (1.65, 0.9243) 
			(1.66, 0.9227) (1.67, 0.9211) (1.68, 0.9196) (1.69, 0.9180) (1.70, 0.9163) (1.71, 0.9147) 
			(1.72, 0.9131) (1.73, 0.9114) (1.74, 0.9097) (1.75, 0.9081) (1.76, 0.9064) (1.77, 0.9046) 
			(1.78, 0.9029) (1.79, 0.9011) (1.80, 0.8994) (1.81, 0.8976) (1.82, 0.8958) (1.83, 0.8940) 
			(1.84, 0.8921) (1.85, 0.8902) (1.86, 0.8884) (1.87, 0.8865) (1.88, 0.8845) (1.89, 0.8826) 
			(1.90, 0.8806) (1.91, 0.8786) (1.92, 0.8766) (1.93, 0.8746) (1.94, 0.8725) (1.95, 0.8705) 
			(1.96, 0.8683) (1.97, 0.8662) (1.98, 0.8640) (1.99, 0.8619) (2.00, 0.8596) (2.01, 0.8574) 
			(2.02, 0.8551) (2.03, 0.8528) (2.04, 0.8505) (2.05, 0.8481) (2.06, 0.8457) (2.07, 0.8432) 
			(2.08, 0.8407) (2.09, 0.8382) (2.10, 0.8356) (2.11, 0.8330) (2.12, 0.8304) (2.13, 0.8277) 
			(2.14, 0.8249) (2.15, 0.8222) (2.16, 0.8193) (2.17, 0.8164) (2.18, 0.8134) (2.19, 0.8104) 
			(2.20, 0.8074) (2.21, 0.8042) (2.22, 0.8010) (2.23, 0.7977) (2.24, 0.7943) (2.25, 0.7909) 
			(2.26, 0.7873) (2.27, 0.7837) (2.28, 0.7799) (2.29, 0.7761) (2.30, 0.7721) (2.31, 0.7680) 
			(2.32, 0.7638) (2.33, 0.7594) (2.34, 0.7548) (2.35, 0.7500) (2.36, 0.7450) (2.37, 0.7398) 
			(2.38, 0.7342) (2.39, 0.7284) (2.40, 0.7221) (2.41, 0.7154) (2.42, 0.7081) (2.43, 0.7000) 
			(2.44, 0.6909) (2.45, 0.6802) (2.46, 0.6667) (2.47, 0.6455)
		};
		\addlegendentry{\small Actual voltage}
		\label{pgfplots:actual_v}
		\addplot[color=blue, dashed, very thick] coordinates {
			(1.00, 0.9937) (1.01, 0.9924) (1.02, 0.9911) (1.03, 0.9898) (1.04, 0.9884) (1.05, 0.9871) 
			(1.06, 0.9857) (1.07, 0.9844) (1.08, 0.9830) (1.09, 0.9816) (1.10, 0.9802) (1.11, 0.9788) 
			(1.12, 0.9773) (1.13, 0.9759) (1.14, 0.9744) (1.15, 0.9729) (1.16, 0.9714) (1.17, 0.9699) 
			(1.18, 0.9684) (1.19, 0.9669) (1.20, 0.9653) (1.21, 0.9638) (1.22, 0.9622) (1.23, 0.9606) 
			(1.24, 0.9590) (1.25, 0.9574) (1.26, 0.9557) (1.27, 0.9541) (1.28, 0.9524) (1.29, 0.9507) 
			(1.30, 0.9490) (1.31, 0.9473) (1.32, 0.9455) (1.33, 0.9438) (1.34, 0.9420) (1.35, 0.9402) 
			(1.36, 0.9384) (1.37, 0.9365) (1.38, 0.9347) (1.39, 0.9328) (1.40, 0.9309) (1.41, 0.9290)
			(1.42, 0.9271) (1.43, 0.9251) (1.44, 0.9232) (1.45, 0.9212) (1.46, 0.9191) (1.47, 0.9171) 
			(1.48, 0.9150) (1.49, 0.9129) (1.50, 0.9108) (1.51, 0.9087) (1.52, 0.9065) (1.53, 0.9044) 
			(1.54, 0.9022) (1.55, 0.8999) (1.56, 0.8977) (1.57, 0.8954) (1.58, 0.8930) (1.59, 0.8907) 
			(1.60, 0.8883) (1.61, 0.8859) (1.62, 0.8835) (1.63, 0.8810) (1.64, 0.8785) (1.65, 0.8760) 
			(1.66, 0.8734) (1.67, 0.8708) (1.68, 0.8682) (1.69, 0.8655) (1.70, 0.8628) (1.71, 0.8601) 
			(1.72, 0.8573) (1.73, 0.8544) (1.74, 0.8516) (1.75, 0.8486) (1.76, 0.8457) (1.77, 0.8427) 
			(1.78, 0.8396) (1.79, 0.8365) (1.80, 0.8333) (1.81, 0.8301) (1.82, 0.8268) (1.83, 0.8234)
			(1.84, 0.8200) (1.85, 0.8166) (1.86, 0.8130) (1.87, 0.8094) (1.88, 0.8057) (1.89, 0.8019) 
			(1.90, 0.7980) (1.91, 0.7941) (1.92, 0.7900) (1.93, 0.7858) (1.94, 0.7815) (1.95, 0.7771) 
			(1.96, 0.7725) (1.97, 0.7678) (1.98, 0.7630) (1.99, 0.7579) (2.00, 0.7527) (2.01, 0.7472) 
			(2.02, 0.7414) (2.03, 0.7353) (2.04, 0.7289) (2.05, 0.7220) (2.06, 0.7146) (2.07, 0.7065) 
			(2.08, 0.6975) (2.09, 0.6872) (2.10, 0.6748) (2.11, 0.6579)
		};
		\addlegendentry{\small Proposed bound}
		\label{pgfplots:prop_bound}
		\addplot[color=red, dotted, very thick] coordinates {
			(1.00, 0.8183) (1.01, 0.8145) (1.02, 0.8105) (1.03, 0.8065) (1.04, 0.8024) (1.05, 0.7983) 
			(1.06, 0.7941) (1.07, 0.7899) (1.08, 0.7855) (1.09, 0.7811) (1.10, 0.7766) (1.11, 0.7720) 
			(1.12, 0.7673) (1.13, 0.7625) (1.14, 0.7577) (1.15, 0.7527) (1.16, 0.7475) (1.17, 0.7423) 
			(1.18, 0.7369) (1.19, 0.7314) (1.20, 0.7257) (1.21, 0.7198) (1.22, 0.7137) (1.23, 0.7074) 
			(1.24, 0.7008) (1.25, 0.6940) (1.26, 0.6868) (1.27, 0.6793) (1.28, 0.6714) (1.29, 0.6629) 
			(1.30, 0.6538) (1.31, 0.6439) (1.32, 0.6330) (1.33, 0.6205) (1.34, 0.6058) (1.35, 0.5866) 
		};
		\addlegendentry{\small Bound by \cite{WangC16}}
		\label{pgfplots:wang_bound}
		\addplot[color=green, dotted, very thick] coordinates {
			(1.00, 0.8296) (1.01, 0.8259) (1.02, 0.8222) (1.03, 0.8184) (1.04, 0.8146) (1.05, 0.8107) (1.06, 0.8067)
			(1.07, 0.8027) (1.08, 0.7987) (1.09, 0.7945) (1.10, 0.7903) (1.11, 0.7860) (1.12, 0.7817) (1.13, 0.7772)
			(1.14, 0.7727) (1.15, 0.7681) (1.16, 0.7633) (1.17, 0.7585) (1.18, 0.7536) (1.19, 0.7485) (1.20, 0.7433)
			(1.21, 0.7380) (1.22, 0.7325) (1.23, 0.7268) (1.24, 0.7209) (1.25, 0.7149) (1.26, 0.7086) (1.27, 0.7021)
			(1.28, 0.6953) (1.29, 0.6881) (1.30, 0.6806) (1.31, 0.6726) (1.32, 0.6641) (1.33, 0.6549) (1.34, 0.6448)
			(1.35, 0.6336) (1.36, 0.6208) (1.37, 0.6052) (1.38, 0.5837) 
		};
		\addlegendentry{\small Bound by \cite{Dvijotham18}}
		\label{pgfplots:dj_bound}
		\end{axis}
		\end{tikzpicture}
		\caption{Comparison of voltage lower bounds at bus 4 of IEEE 39-bus system given by Theorem \ref{thm:main} (shown by \ref{pgfplots:prop_bound}) and that given in \cite{WangC16, Dvijotham18} (shown by \ref{pgfplots:wang_bound} and \ref{pgfplots:dj_bound}, repectively) as system load powers build up. Actual voltage profile is shown by \ref{pgfplots:actual_v}. Voltage bound estimations cease to exist when the existence of power flow solutions cannot be certified by the corresponding methods.} \label{fig:exp3c}
	\end{figure*}
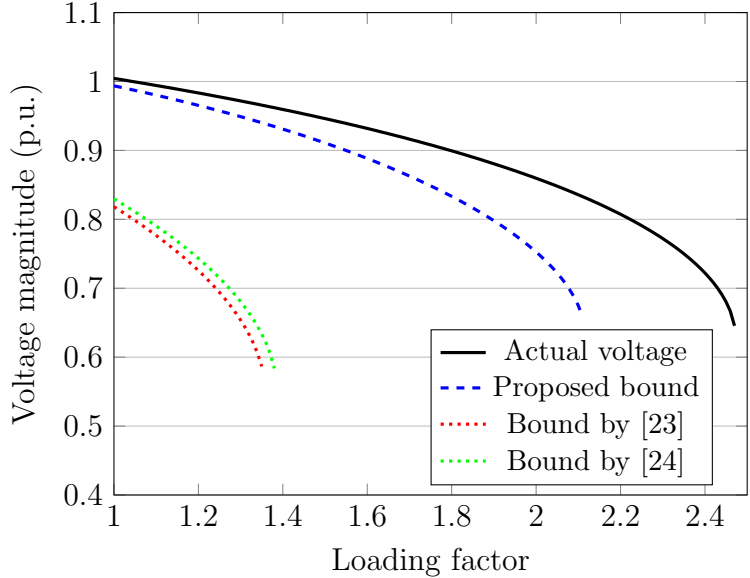
	
	\section{Discussion and Future Directions}
	
	We have presented a strengthened power flow solvability condition for large-scale power systems. The conservativeness issue in existing conditions has been significantly improved --- with negligible computational overhead, the condition provides much tighter lower bound of power flow solvability limit compared to existing ones. Thanks to the exploitation of properties of power flow equations in complex space, uniqueness of power flow solution in state space can be guaranteed for a wider range of power injections. As we show in Supplementary Information, this also ensures that the unique power flow solution can be obtained by iterating the fixed-point power flow equations. The proposed condition in Theorems \ref{thm:main} can help improve situational awareness of system operators by providing quick system stress assessment and critical area identification based on the interplay between load power injections and the normalized impedance matrix encoding system parametric and topological information. The real part of the vector inner product term in \eqref{eq:gamma} is novel, which strengthens the condition and also consolidates physical insights of the role load power factor plays in system long-term stability and power flow solvability.
	
	Some future research directions include extension of the condition to more realistic transmission system models. Specifically,  we would like to relax the assumption of constant generator voltage phasors. Although this assumption is widely adopted in power system steady-state stability analysis and works quite well under normal operating conditions, it may break down when systems are close to their steady-state stability limits. Results considering generators with varying phase angles (PV bus model) have been reported, for example, in \cite{Simpson17a, Simpson17b, Dorfler13}, but they are quite conservative and/or are restricted to systems under certain modeling assumptions. Another important future work is to investigate the applicability of the analytical tool to other system models. For instance, we are looking into ways to develop similar index for unbalanced three-phase distribution systems with transformers and other component models. Existing works along the line include \cite{Bernstein18,WangC17}. However, we believe the analytical tool we developed in this paper will facilitate the derivation of sharper solvability condition and provide novel physical insights into the problem. In addition, we expect similar approaches can be applied to analyze other complex infrastructure networks, such as water distribution systems \cite{Bazrafshan18}, natural gas systems \cite{Dvijotham15}, as well as their interconnections with electric power systems \cite{Shahidehpour05}.

	\newpage
	\appendix
	\section*{Supplementary Information}
	
	\section{Background} \label{sect2}
	
	\subsection{Notations}
	
	\subsubsection{Sets}
	$\mathbb{R}$, $\mathbb{R}_{++}$, and $\mathbb{C}$ are the set of real, positive real, and complex numbers, respectively. The disk in $\mathbb{C}$ with center $c \in \mathbb{C}$ and radius $r \in \mathbb{R}_{++}$ is denoted by $\mathbb{D}(c; r)$, that is, $\mathbb{D}(c; r) := \{ z \in \mathbb{C} :\; |z - c| < r \}$. The unit disk $\mathbb{D}((0,0); 1)$ is shorthanded as $\mathbb{D}$. Given $\mb{c} \in \mathbb{C}^n$ and $\mb{r} \in \mathbb{R}_{++}^n$, $\mathbb{D}^n(\mb{c}; \mb{r})$ is the $n$-dimensional polydisc defined as $\mathbb{D}^n(\mb{c}; \mb{r}) := \mathbb{D}(c_1; r_1) \times \cdots \times \mathbb{D}(c_n; r_n)$. The closure, interior, and boundary of a set $S$ are denoted by $\bar{S}$, $\inte(S)$, and $\partial S$, respectively.
	
	\subsubsection{Vectors and matrices} 
	Vectors and matrices are represented by bold letters while scalars are represented by normal ones. Let $\mb{e}_i^n$ be the $i$th canonical basis vector of $\mathbb{R}^n$, that is, the $i$th entry of $\mb{e}_i^n$ is $1$ and all other entries are $0$. For matrix $\mb{A} \in \mathbb{C}^{m\times n}$, $\mb{A}^\top$, $\mb{A}^H$ are respectively the transpose and conjugate transpose of $\mb{A}$. $a_i$ denotes the vector formed by the $i$th row of $\mb{A}$. For vector $\mb{x} \in \mathbb{C}^n$, $\| \mb{x} \|_p$ denotes the $\ell_p$ norm of $\mb{x}$ where $p \in [1, \infty) \cup \{\infty\}$ and $\diag(\mb{x}) \in \mathbb{C}^{n \times n}$ is the associated diagonal matrix. Unless otherwise stated, $|\mb{x}|$ denotes the $\ell_\infty$ norm. $\mb{0}$ and $\mb{1}$ are the vectors of all $0$'s and $1$'s of appropriate sizes. The cardinality of a set or the absolute value of a (possibly) complex number is denoted by $|\cdot|$. $\mathrm{i} = \sqrt{-1}$ is the imaginary unit. $\mathbf{I}_n$ denotes identity matrix of appropriate dimension. The real and imaginary parts of a complex number are denoted by $\re(\cdot)$ and $\im(\cdot)$, respectively.
	
	\subsection{System modeling}
	
	\subsubsection{Generator model}
	
	Synchronous generators are a primary source of power supply and are to a great extend responsible for maintaining proper voltage profile across power system through automatic voltage regulators \cite[Chap.~3]{Cutsem08}. Under normal operating conditions they are capable of maintaining constant real power outputs and voltage magnitudes. Therefore the generator buses are generally modeled as the so-called PV buses where the bus real power injections and voltage magnitudes are specified while the reactive power injections and voltage phase angles are not regulated. Under extreme conditions when excessive current flows through field winding of the synchronous generator, the overexcitation limiter takes effect, which results in loss of the capability to regulate voltage magnitude and output of incremental reactive power. A common modeling practice under this scenario is to switch the bus type from PV to PQ, thus fixing the real and reactive power injection of the bus while allowing voltage magnitude to vary. It is remarked that in almost all voltage instability incidents some generators were operating with limited reactive capability \cite[Chap.~3]{Cutsem08}.
	
	For practical power systems, the generator buses have regulated voltage magnitudes and small phase angles. It is common in voltage stability analysis to assume that the generator buses have constant voltage phasors \cite{Vu99, WangY11}. We adopt this assumption in this paper, and note that it can be partially justified by the fact that voltage instability/power flow insolvability are mostly caused by system overloading due to excess demand at load side, not the generator side.
	
	\subsubsection{Load model}
	
	We adopt the simple constant power load model in the study where the real and reactive power demand are known and specified for each load bus. More general load model can be incorporated, for instance the ZIP static load model, where the constant impedance (Z), constant current (I), and constant power (P) characteristics of the load are simultaneously taken into account. From a modeling perspective, constant current and impedance loads can be `adsorbed' by the system admittance matrix and normalized impedance matrix \cite{Wang17}, thus there is little generality lost when only considering constant power loads. Since our emphasis is on the investigation of voltage-power relationship of static power flow equations, dynamic loads are not considered. 
	
	\subsection{Power flow equations}
	
	In this section, we slightly generalize the power flow equations $\mb{v}_L = \mb{1} - \hat{\mb{Z}}\diag^{-1}(\mb{v}_L^*)\mb{S}_L^*$ in \eqref{eq:fppf} by incorporating in the power flow equations a known power flow solution $\mb{v}^0_L$ and its associated load power $\mb{S}^0_L$. All conditions we derived in the sequel uses this strengthened version of power flow equation. This allows the certificate of solution existence and uniqueness for incremental load powers around some nominal ones, which is a typical use case in many applications. We see that $\mb{v}^0_L = \mb{1}$ is a power flow solution for $\mb{S}^0_L = \mb{0}$, so the requirement of the existence of a known power flow solution does not affect the generality of the condition.
	
	The power flow equation with load $\mb{S}_L := \mb{S}^0_L + \mb{\sigma}_L$ is
	\begin{equation} \label{eq:pfv}
	\mb{v}_L = \mb{1} - \hat{\mb{Z}} \diag^{-1}(\mb{v}_L^*)\mb{S}_L^*.
	\end{equation}
	In addition, since $\mb{v}^0_L$ is a power flow solution for load $\mb{S}^0_L$, we have
	\begin{equation} \label{eq:pfv0}
	\mb{v}^0_L = \mb{1} - \hat{\mb{Z}} \diag^{-1}\left({\mb{v}_L^0}^*\right) {\mb{S}_L^0}^*.
	\end{equation}
	Substitute $\mb{1}$ in \eqref{eq:pfv0} back into \eqref{eq:pfv}, we get
	\begin{subequations} \label{eq:pfv01}
		\begin{align}
		\mb{v}_L &= \mb{v}^0_L + \hat{\mb{Z}} \diag^{-1}\left({\mb{v}^0_L}^*\right) {\mb{S}^0_L}^* - \hat{\mb{Z}} \diag^{-1}(\mb{v}_L^*) \mb{S}_L^* \\
		&= \mb{v}^0_L - \hat{\mb{Z}} \diag^{-1}\left({\mb{v}^0_L}^*\right) \mb{\sigma}_L^* + \hat{\mb{Z}} \left( \diag^{-1}(\mb{v}^0_L) - \diag^{-1}(\mb{v}_L)\right)^* \mb{S}_L^*. \label{eq:pfv02}
		\end{align}
	\end{subequations}
	Apply change of variable $\mb{u} := \diag^{-1}(\mb{v}^0_L)\mb{v}_L$ and left multiply both sides of \eqref{eq:pfv02} by $\diag^{-1}(\mb{v}^0_L)$ yields
	\begin{equation} \label{eq:pfu}
	\mb{u} = F(\mb{u}) := \mb{1} - \tilde{\mb{Z}} \mb{\sigma}_L^* + \tilde{\mb{Z}} \left( \mathbf{I} - \diag^{-1}(\mb{u}^*) \right) \mb{S}_L^*
	\end{equation}
	where the voltage-normalized impedance is defined as $\tilde{\mb{Z}} = \diag^{-1}(\mb{v}_L^0) \hat{\mb{Z}} \diag^{-1}({\mb{v}_L^0}^*)$. We work with the power flow equations \eqref{eq:pfu} in the sequel.
	
	\subsection{Some system theoretical quantities}
	
	Now we introduce some system stress measures that will be used in deriving solvability conditions. For load bus $i \in \mathcal{N}_L$, the quantities
	\begin{subequations}
		\begin{align}
		\eta_i(\mb{\sigma}_L) &:= \tilde{z}_i^\top \mb{\sigma}^*_L, \\ \xi_i(\mb{S}_L) &:= \left\| \tilde{z}_i^\top \diag(\mb{S}^*_L) \right\|_1
		\end{align}
	\end{subequations}
	quantify nodal stress levels resulted from incremental and total loads. They appear in existing solvability literature \cite{WangC16,Dvijotham18,Nguyen17}. In addition to the two stress measures above, we introduce an additional one fusing the two as
	\begin{equation}
	\gamma_i(\mb{S}_L, \mb{\sigma}_L) = 2\left(\xi_i(\mb{S}_L) + \re\left( \eta_i(\mb{\sigma}_L) \right) \right) - \xi_i(\mb{S}_L)^2 - \left| \eta_i(\mb{\sigma}_L) \right|^2.
	\end{equation}
	The solvability conditions are given in terms of the maxima of $\xi_i$, $|\eta_i|$, and $\gamma_i$ over the set of load buses. We denote the maxima of the corresponding quantities as
	\begin{subequations}
		\begin{align}
		\gamma(\mb{S}_L, \mb{\sigma}_L) &:= \max_{i \in \mathcal{N}_L} \gamma_i(\mb{S}_L, \mb{\sigma}_L), \\
		\eta(\mb{\sigma}_L) &:=  \max_{i \in \mathcal{N}_L} \left| \eta_i(\mb{\sigma}_L) \right| = \left\| \tilde{z}_i^\top \mb{\sigma}^*_L \right\|_\infty, \\
		\xi(\mb{S}_L) &:= \max_{i \in \mathcal{N}_L} \xi_i(\mb{S}_L) = \left\|\tilde{z}_i^\top \diag(\mb{S}^*_L) \right\|_\infty.
		\end{align}
	\end{subequations}
	
	Next we define the following set in $\mathbb{C}^n$ parameterized by $r \in \mathbb{R}_{++}$. It is an $n$-dimensional polydisc when $\xi_i(\mb{S}_L)$ are strictly positive for all $i\in \mathcal{N}_L$
	\begin{equation} \label{eq:Ir}
	\mathcal{D}(r) = \left\{ \mb{u} \in \mathbb{C}^n :\; \begin{cases} \left| 1 - \eta_i(\mb{\sigma}_L) - u_i \right| < r \xi_i(\mb{S}_L), & \xi_i(\mb{S}_L) > 0 \\
	u_i = 1 - \eta_i(\mb{\sigma}_L), & \xi_i(\mb{S}_L) = 0
	\end{cases} \right\}.
	\end{equation}
	We will derive conditions under which $\bar{\mathcal{D}}(r)$ is an invariant set for power flow mapping $F(\mb{u})$.

	\section{Existence of Power Flow Solutions}
	
	In this section, we derive sufficient condition guaranteeing the existence of solutions to fixed point power flow equations \eqref{eq:pfu}. We first introduce a basic result, the Brouwer fixed point theorem due to L. E. J. Brouwer, which establishes the existence of fixed point for equations in Euclidean space of the form $x = f(x)$.
	
	\begin{thm}[Brouwer fixed point theorem] \label{thm:Brouwer}
		Let $f : C \subset E^n \to E^n$ be continuous on the compact, convex set $C$ which is a subset of the $n$-dimensional Euclidean space $E^n$, and suppose that $f(C) \subseteq C$. Then $f$ has a fixed point in $C$.
	\end{thm}
	
	Before applying Theorem \ref{thm:Brouwer} to derive sufficient condition for the existence of fixed points for power flow equations \eqref{eq:pfu}, we provide sufficient condition on the existence of positive $r$ such that $\mathcal{D}(r)$ is an invariant set for \eqref{eq:pfu} in Lemma \ref{thm:invariant} below.
	
	\begin{lem} \label{thm:invariant}
		Given load power $\mb{S}_L^0$ and power flow solution $\mb{v}_L^0$ satisfying \eqref{eq:pfv0} and denote
		\begin{equation} \label{eq:Delta}
		\Delta := \left( 1 - \gamma(\mb{S}_L, \mb{\sigma}_L) \right)^2 - 4 \xi^2(\mb{S}_L) \eta^2(\mb{\sigma}_L),
		\end{equation}
		then the power flow mapping $\mb{u} = F(\mb{u})$ defined in \eqref{eq:pfu} with load power $\mb{S}_L = \mb{S}_L^0 + \mb{\sigma}_L$ maps the closure of $\mathcal{D}(r)$ to itself, that is, $F(\bar{\mathcal{D}}(r)) \subseteq \mathcal{D}(r)$, for
		\begin{equation} \label{eq:rbound}
		\begin{cases}
		r \in \left( \sqrt{ \frac{1 - \gamma(\mb{S}_L, \mb{\sigma}_L) - \sqrt{\Delta}}{2 \xi(\mb{S}_L)^2} }, \sqrt{ \frac{1 - \gamma(\mb{S}_L, \mb{\sigma}_L) + \sqrt{\Delta}} {2 \xi(\mb{S}_L)^2} } \right), & \xi(\mb{S}_L) > 0 \\
		r > 0, & \xi(\mb{S}_L) = 0,
		\end{cases}
		\end{equation} 
		when the following two conditions hold
		\begin{subequations} \label{eq:cond}
			\begin{align}
			\gamma(\mb{S}_L, \mb{\sigma}_L) + 2 \xi(\mb{S}_L) \eta(\mb{\sigma}_L) < 1 \label{eq:cond_bl}, \\
			\xi(\mb{S}_L) - \eta(\mb{\sigma}_L) \le 1 \label{eq:condb}.
			\end{align}
		\end{subequations}
	\end{lem}
	\begin{proof}
		When $\xi_i(\mb{S}_L) = 0$, $\mb{S}_L = \mb{0}$ since $\tilde{\mb{Z}}$ is full rank, so the $i$th dimension of $\mathcal{D}(r)$ degenerates to a point $u_i = 1 - \tilde{z}_i^\top \mb{\sigma}_L^*$ and $F_i(\mb{u}) = 1 - \tilde{z}_i^\top \mb{\sigma}_L^*$ for all $\mb{u} \in \{\mb{u} \in \mathbb{C}^n \mid u_i \ne 0, i \in \mathcal{N}_L \}$, so the lemma trivially holds for these dimensions. Hence, we may assume for the rest of the proof that $\xi_i(\mb{S}_L) \ne 0$ for all $i \in \mathcal{N}_L$.
		
		We know that for $z \in \mathbb{C}$ and $r \in \mathbb{R}_{++}$, the set of $z$ characterized by the inequality $|z^* - 1| / |z^*| < r$ can be a ball ($r < 1$), a half plane ($r = 1$), or the complement of a closed ball ($r > 1$):
		\begin{equation} \label{eq:ineq_r}
		\left\{ z \in \mathbb{C} :\; \left| \frac{z^*-1}{z^*} \right| < r \right\} =
		\begin{cases}
		\left\{ z \in \mathbb{C} :\; \left| z - \frac{1}{1-r^2} \right| < \frac{r}{1-r^2} \right\}, & r < 1, \\
		\left\{ z \in \mathbb{C} :\; \re(z) > 1/2 \right\}, & r = 1, \\
		\left\{ z \in \mathbb{C} :\; \left| z + \frac{1}{r^2-1} \right| > \frac{r}{r^2-1} \right\}, & r > 1.
		\end{cases}
		\end{equation}
		Define the $n$-dimensional analogy of the set \eqref{eq:ineq_r} as $\mathcal{U}(r) := \left\{ \mb{z} \in \mathbb{C}^n :\; \left| z_i^*-1\right| / \left|z_i^*\right| < r \right\}$. Assume that each diagonal entry $(u_i^*-1)/u_i^*$ of the diagonal matrix $\mathbf{I} - \diag^{-1}(\mb{u}^*)$ in \eqref{eq:pfu} lies between $(-r, r)$ for some $r > 0$, then we have $1 - \tilde z_i^\top \mb{\sigma}_L^* + \tilde z_i^\top (\mathbf{I} - \diag^{-1}(\mb{u}^*))\mb{S}_L^* \subseteq \mathcal{D}(r)$, which means $F( \mathcal{U}(r) ) \subseteq \mathcal{D}(r)$. Therefore, to show $F(\bar{\mathcal{D}}(r)) \subseteq \mathcal{D}(r)$ for some $r > 0$, we simply need to show $\bar{\mathcal{D}}(r) \subseteq \mathcal{U}(r)$ for the given $r$. Based on \eqref{eq:ineq_r}, we discuss in three distinct cases depending on whether $r < 1$, $r > 1$, or $r = 1$.
		
		First, we consider the case when $r < 1$. Based on \eqref{eq:ineq_r}, the condition $\bar{\mathcal{D}}(r) \subseteq \mathcal{U}(r)$ simply indicates that for each $i$, the closed ball centered at $1-\eta_i(\mb{\sigma}_L)$ with radius $r\xi_i(\mb{S}_L)$ is contained in the open ball centered at $(1/(1-r^2), 0)$ with radius $r/(1-r^2)$, which is equivalent to the condition that the distance between the two centers is less than the difference of their radii: 
		\begin{equation} \label{eq:ineq6}
		\left|1 - \eta_i(\mb{\sigma}_L) - \frac{1}{1-r^2}\right| < \frac{r}{1-r^2} - r\xi_i(\mb{S}_L), \qquad \forall i.
		\end{equation}
		
		Multiply $(1 - r^2)$ and square both sides of the inequality, and note the right hand side has to be positive, we obtain the following equivalent representation of \eqref{eq:ineq6}:
		\begin{subequations} \label{eq:multsq}
			\begin{align}
			\xi_i(\mb{S}_L)^2r^4 + (\gamma_i(\mb{S}_L, \mb{\sigma}_L) - 1)r^2 + |\eta_i(\mb{\sigma}_L)|^2 &< 0, \qquad \forall i, \label{eq:multsqa} \\
			\xi_i(\mb{S}_L)(1 - r^2) - 1 &< 0, \qquad \forall i. \label{eq:multsqb}
			\end{align}	
		\end{subequations}
		
		Next, we consider the case when $r > 1$. The argument is very similar: based on \eqref{eq:ineq_r}, the condition $\bar{\mathcal{D}}(r) \subseteq \mathcal{U}(r)$ indicates that for each $i$, the closed ball centered at $1-\eta_i(\mb{\sigma}_L)$ with radius $r\xi_i(\mb{S}_L)$ lies outside the open ball centered at $(1/(1-r^2), 0)$ with radius $r/(r^2-1)$, which is equivalent to the condition that the distance between the two centers is greater than the sum of their radii: 
		\begin{equation}
		\left|1 - \eta_i(\mb{\sigma}_L) + \frac{1}{r^2-1}\right| > \frac{r}{r^2-1} + r\xi_i(\mb{S}_L), \qquad \forall i,
		\end{equation} 
		After simplifications, we obtain the same inequality as \eqref{eq:multsqa} (note there is no counterpart for \eqref{eq:multsqb} since $r/(r^2-1) + r\xi_i(\mb{S}_L)$ is always positive).
		
		Lastly, when $r = 1$, the condition $\bar{\mathcal{D}}(r) \subseteq \mathcal{U}(r)$ is satisfied when the closed ball centered at $1-\eta_i(\mb{\sigma}_L)$ with radius $\xi_i(\mb{S}_L)$ lies in the half plane $\{ u \in \mathbb{C} \mid \re(u) > 1/2 \}$ for every $i$, which is
		\begin{equation}
		\xi_i(\mb{S}_L) + \re(\eta_i(\mb{\sigma}_L)) < 1/2, \qquad \forall i.
		\end{equation}
		However, it is easy to verify that this condition is identical to \eqref{eq:multsqa} when $r = 1$.
		
		In summary, we have shown that there exists $r > 0$ such that $\bar{\mathcal{D}}(r) \subseteq \mathcal{U}(r)$ if and only if there exists $r > 0$ such that \eqref{eq:multsq} holds (note that \eqref{eq:multsqb} always holds when $r \ge 1$). Since $\xi(\mb{S}_L)$, $\eta(\mb{\sigma}_L)$, and $\gamma(\mb{S}_L, \mb{\sigma}_L)$ are the maxima of the corresponding quantities over all load bus $i$, \eqref{eq:multsqa} is implied by
		\begin{equation} \label{eq:b10}
		\xi(\mb{S}_L)^2r^4 + (\gamma(\mb{S}_L, \mb{\sigma}_L) - 1)r^2 + \eta(\mb{\sigma}_L)^2 < 0.
		\end{equation}
		Therefore, to prove the lemma, we only need to show \eqref{eq:cond} implies \eqref{eq:b10} and \eqref{eq:multsqb} for some $r > 0$. Condition \eqref{eq:b10} is a quadratic inequality in $r^2$ and it can be easily checked that condition \eqref{eq:cond_bl} implies $r^2 = \frac{1}{2} \left( 1-\gamma(\mb{S}_L, \mb{\sigma}_L) \right) / \xi(\mb{S}_L)^2 > 0$ satisfies \eqref{eq:b10}. When $\left( 1 - \gamma(\mb{S}_L, \mb{\sigma}_L) \right) / \xi(\mb{S}_L)^2 \ge 2$, \eqref{eq:multsqb} always holds. Otherwise it is implied by the following inequality:
		\begin{equation} \label{eq:b11}
		\xi(\mb{S}_L) \left(1 - \frac{1-\gamma(\mb{S}_L, \mb{\sigma}_L)}{2\xi(\mb{S}_L)^2}\right) - 1 < 0
		\end{equation}
		since we replace each $\xi_i(\mb{S}_L)$ in \eqref{eq:multsqb} by $\xi(\mb{S}_L)$. Condition \eqref{eq:b11} can be rewritten as $\gamma(\mb{S}_L, \mb{\sigma}_L) + 2\xi(\mb{S}_L)\left( \xi(\mb{S}_L) - 1 \right) < 1$, which is implied by \eqref{eq:cond}. We have thus shown condition \eqref{eq:cond} implies $\bar{\mathcal{D}}(r) \subseteq \mathcal{U}(r)$, which then implies $F(\bar{\mathcal{D}}(r)) \subseteq \mathcal{D}(r)$ as we have mentioned above. In addition, the bounds on $r$ in \eqref{eq:rbound} are simply the square roots of the solutions to the quadratic equation corresponding to \eqref{eq:b10}.
	\end{proof}

	Combining Brouwer fixed point theorem (Theorem \ref{thm:Brouwer}) and Lemma \ref{thm:invariant}, we arrive at the main result of this section --- existence of power flow solutions:
	
	\begin{thm}[Existence of power flow solutions] \label{thm:exist}
		Given load power $\mb{S}_L^0$ and power flow solution $\mb{v}_L^0$ satisfying \eqref{eq:pfv0} and define $\Delta$ as in \eqref{eq:Delta},
		then the power flow equation $\mb{u} = F(\mb{u})$ defined in \eqref{eq:pfu} with load power $\mb{S}_L = \mb{S}_L^0 + \mb{\sigma}_L$ admits at least one solution in $\bar{\mathcal{D}}(r)$ where
		\begin{equation} \label{eq:rdef}
		\begin{cases} 
		r = \sqrt{ \frac{1 - \gamma(\mb{S}_L, \mb{\sigma}_L) - \sqrt{\Delta}} {2 \xi(\mb{S}_L)^2} }, & \xi(\mb{S}_L) > 0,  \\ 
		r > 0, & \xi(\mb{S}_L) = 0,
		\end{cases}
		\end{equation}
		when condition \eqref{eq:cond} holds.
	\end{thm}
	\begin{proof}
		When $\xi_i(\mb{S}_L) = 0$, the $i$th dimension of $\mathcal{D}(r)$ degenerates to a point $u_i = 1 - \tilde{z}_i^\top \mb{\sigma}_L^*$ and $F_i(\mb{u}) = 1 - \tilde{z}_i^\top \mb{\sigma}_L^*$ for all $\{\mb{u} \in \mathbb{C}^n \mid u_i \ne 0\}$, so the theorem trivially holds for these dimensions. Hence, we may assume for the rest of the proof that $\xi_i(\mb{S}_L) \ne 0$ for all $i$.
		
		By Lemma \ref{thm:invariant}, $\bar{\mathcal{D}}(r)$ where $r$ satisfies \eqref{eq:rbound} is a compact and convex invariant set for $F(\mb{u})$ when condition \eqref{eq:cond} holds, so we know from Theorem \ref{thm:Brouwer} that the power flow equation $\mb{u} = F(\mb{u})$ admits a solution in $\bar{\mathcal{D}}(r)$. Since $F$ maps $\bar{\mathcal{D}}(r)$ to $\mathcal{D}(r)$, the solution has to lie in $\mathcal{D}(r)$. If we denote the infimum and supremum of $r$ in \eqref{eq:rbound} by $\ubar{r}$ and $\bar{r}$, then it follows that $\mb{u} = F(\mb{u})$ admits a solution in $\bigcap_{\ubar{r}<r<\bar{r}} \mathcal{D}(r) = \bar{\mathcal{D}}(\ubar{r})$.
	\end{proof}

	The following proposition exploits some implications of Theorem \ref{thm:exist} when no power flow solutions are available \emph{a priori}, i.e., when $\mb{S}^0_L = \mb{0}$. See Main Text for further discussions.
	
	\begin{pro} \label{thm:increase}
		Given a vector of load powers $\mb{s}_L$ such that $\xi(\mb{s}_L) - \eta(\mb{s}_L) = 1$, the scalar function $f(\lambda) := \gamma(\lambda\mb{s}_L) + 2\xi(\lambda\mb{s}_L)\eta(\lambda\mb{s}_L)$ is increasing on $\lambda \in [0,1]$ and $f(1) \ge 1$.
	\end{pro}
	\begin{proof}
		To prove the proposition, we only need to show the functions $f_i(\lambda) = \gamma_i(\lambda \mb{s}_L) + 2\xi(\lambda\mb{s}_L)\eta(\lambda\mb{s}_L)$ are increasing on $\lambda \in [0, 1]$ for all $i \in \mathcal{N}_L$, as $f(\lambda)$ is simply the maximum of $f_i(\lambda)$. Since $\xi(\lambda\mb{s}_L) = \lambda \xi(\mb{s}_L)$ and $\eta(\lambda\mb{s}_L) = \lambda \eta$, we have 
		\begin{equation}
		f_i(\lambda) = - \left( \xi_i(\mb{s}_L)^2 + |\eta_i(\mb{s}_L)|^2 - 2\xi(\mb{s}_L)\eta(\mb{s}_L) \right)\lambda^2 + 2\left( \xi_i(\mb{s}_L) + \re(\eta_i(\mb{s}_L)) \right)\lambda.
		\end{equation}
		There are three cases to consider depending on whether $\xi_i(\mb{s}_L)^2 + |\eta_i(\mb{s}_L)|^2 - 2\xi(\mb{s}_L)\eta(\mb{s}_L)$ is equal to zero, less than zero, or greater than zero. For notational simplicity, we make the dependence of $\mb{s}_L$ implicit in all functions in the remainder of the proof.
		
		When $\xi_i^2 + |\eta_i|^2 - 2\xi\eta = 0$, $f_i(\lambda) = 2(\xi_i + \re(\eta_i))\lambda$ is increasing since $2(\xi_i + \re(\eta_i)) \ge 0$. When $\xi_i^2 + |\eta_i|^2 - 2\xi\eta < 0$, the axis of symmetry of the parabola $f_i(\lambda) = 0$ is less than or equal to zero and consequently $f_i(\lambda)$ is increasing for $\lambda \ge 0$. When $\xi_i^2 + |\eta_i|^2 - 2\xi\eta > 0$, we need to show the axis of symmetry $\lambda_{\mathrm{as}}$ of $f_i(\lambda) = 0$ is greater than or equal to 1. This is indeed the case as
		\begin{equation} \label{eq:b16}
		\lambda_{\mathrm{as}} = 
		\frac{\xi_i + \re(\eta_i)}{\xi_i^2 + |\eta_i|^2 - 2\xi\eta} 
		\ge \frac{\xi_i - |\eta_i|}{\xi_i\xi + |\eta_i|\eta - \xi_i\eta - |\eta_i|\xi}
		= \frac{\xi_i - |\eta_i|}{(\xi_i - |\eta_i|)(\xi - \eta)} = 1,
		\end{equation}
		as desired. To show $f(1) \ge 1$, let $k = \argmax \xi_i$, then it follows from \eqref{eq:b16} that $\xi_k + \re(\eta_k) \ge \xi_k^2 + |\eta_k|^2 - 2\xi\eta$. Move all terms to the left and add $\xi_k + \re(\eta_k)$ on both sides, we obtain $\gamma_k + 2\xi\eta \ge \xi_k + \re(\eta_k) = \xi + \re(\eta_k) \ge 1$, which implies $f(1) = \gamma + 2\xi\eta \ge 1$.
	\end{proof}

	\section{Uniqueness of Power Flow Solutions and Convergence of Power Flow Iteration} \label{sect:uniqueness}
	
	We show in this section the uniqueness of high-voltage power flow equation. Specifically, in the first subsection, we show the general results on uniqueness of solution to fixed point equations and convergence of the fixed point iteration. The results are then applied to the power flow equations in the second subsection.
	
	\subsection{General theory of uniqueness of fixed point in polydisc and convergence of fixed point iteration} \label{sect:unique}
	
	As opposed to previous approaches \cite{Bolognani16,Simpson-Porco16,WangC16,Simpson17b} which rely on the contraction properties of the form $\|f'(x)\| < 1$ for the power flow equations, we take an alternative route. As noted in \cite{Henrici74}, by making a more efficient use of properties of holomorphic functions, the uniqueness of fixed point can be proved without making explicit contraction conditions on the boundedness of $\|f'(x)\|$.
	
	We have the following standard result in complex analysis:
	\begin{thm}[\hspace{1sp}{\cite[Thm. 6.12a]{Henrici74}}] \label{thm:Henrici}
		Let $f$ be holomorphic in a simply connected region $S \subseteq \mathbb{C}$ and continuous on the closure $\bar{S}$ of $S$, and let $\bar{f}(S)$ be a bounded set contained in $S$. Then $f$ has exactly one fixed point.
	\end{thm}
	
	The uniqueness of fixed point is a direct consequence of Rouch\'e's theorem. To show the uniqueness of fixed point of the power flow equations, we generalize Theorem \ref{thm:Henrici} to functions defined on subsets of $\mathbb{C}^n$ by noting the following generalized Rouch\'e's theorem:
	\begin{thm}[Generalized Rouch\'e's theorem \cite{Lloyd79}] \label{thm:ndimRouche}
		Let $D$ be a bounded, open subset of $\mathbb{C}^n$ and suppose that $f, g$ are continuous functions of $\bar{D}$ into $\mathbb{C}^n$ that are holomorphic in $D$ such that
		\begin{equation}
		|g(\mb{z})| < |f(\mb{z})| \qquad \mb{z} \in \partial D
		\end{equation}
		for some norm $|\cdot|$. Then $f$ has finitely many zeros in $D$, and counting multiplicity, $f$ and $f + g$ have the same number of zeros in $D$.
	\end{thm}
	
	The convergence of fixed point iteration for real-valued functions defined in a complete metric space can generally be shown through Banach fixed point theorem \cite[Thm. 5.1.3]{Ortega70}.
	
	Here, we show that the fixed point iteration of the complex function defined in the last section also converges to the unique fixed point. The proof is, similar to the proof of uniqueness of fixed point in subsection \ref{sect:unique}, also an extension of the result in \cite{Henrici74} to higher dimensions. To pave the way for the proof, we first present a generalization of Schwarz's lemma in several variables:
	
	\begin{lem}[Schwarz's lemma in several variables] \label{thm:ndimSchwarz}
		Suppose $f: \mathbb{C}^n \to \mathbb{C}$ is holomorphic in a neighborhood of $\bar{\mathbb{D}}^n$, $f(\mb{0}) = 0$, further suppose for all $\mb{z} \in \mathbb{D}^n$, $|f(\mb{z})| \le M$ for some $M$, then
		\begin{equation}
		|f(\mb{z})| \le M |\mb{z}|
		\end{equation}
		for all $\mb{z} \in \bar{\mathbb{D}}^n$.
	\end{lem}
	\begin{proof}
		We define the function $g: \mathbb{D} \to U \subset \mathbb{C}$ as $g(s) = f(s\mb{w})$ where $\mb{w} \in \partial \mathbb{D}^n$, then we know $g$ is holomorphic, $g(0) = 0$ and $|g(s)| \le M$ for all $s \in \mathbb{D}$. It follows from Schwarz's lemma that 
		\begin{equation}
		|g(s)| \le M|s|, \quad \forall s \in \mathbb{D},
		\end{equation}
		or
		\begin{equation}
		|f(s\mb{w})| \le M|s| =  M | s\mb{w} |.
		\end{equation}
		Since $\mb{w}$ is arbitrary, the result is thus implied from the above inequality.
	\end{proof}
	
	In addition, we also examine the convergence rate of the power flow iteration. Solving power flow equations is the most fundamental task in power systems analysis. The fixed point iteration introduced above serves as an alternative approach to solve the power flow equations besides the most frequently used Newton-Raphson method. In this section, we discuss the rate of convergence of the fixed point iteration, which is of great practical importance concerning the applicability of the fixed point iteration in solving power flow equations. Specifically, we will show that the fixed point iteration exhibits linear convergence rate.
	
	With the above theorems, we are now ready to present the main result in this section:
	\begin{thm}[Uniqueness of fixed point in $n$-dimension] \label{thm:ndimHenrici}
		Given vectors $\mb{c} \in \mathbb{C}^n$ and $\mb{r} \in \mathbb{R}_{++}^n$, let $f : \bar{\mathbb{D}}^n(\mb{c}, \mb{r}) \to \mathbb{D}^n(\mb{c}, \mb{r})$ be a function holomorphic in $\mathbb{D}^n(\mb{c}, \mb{r})$ and continuous on the closure $\bar{\mathbb{D}}^n(\mb{c}, \mb{r})$, and $\bar{f}(\mathbb{D}^n(\mb{c}, \mb{r}))$ is contained in $\mathbb{D}^n(\mb{c}, \mb{r})$. Then $f$ has exactly one fixed point in $\mathbb{D}^n(\mb{c}, \mb{r})$. Moreover, the sequence $\{ \mb{z}^n \}$ defined as 
		\begin{equation}
		\mb{z}^{n+1} = f(\mb{z}^n), \quad n = 0, 1, 2, \ldots
		\end{equation}
		converges to the unique fixed point $\mb{w}$ given any $\mb{z}^0 \in \mathbb{D}^n(\mb{c}, \mb{r})$ in such a manner that
		\begin{equation}
		|\mb{z}^n - \mb{w}| < |\mb{r}|(1 + \mu) \left( \frac{2\mu}{1 + \mu^2} \right)^n,\quad  n = 0, 1, 2, \ldots
		\end{equation}
		for some number $0 \le \mu < 1$.
	\end{thm}
	\begin{proof}
		
		We first consider the case in which $\mb{c} = \mb{0}$ and $\mb{r} = \mb{1}$ such that $\mathbb{D}^n(\mb{c}, \mb{r}) = \mathbb{D}^n$ is the $n$-dimensional unit polydisc. The condition $\bar{f}(\mathbb{D}^n) \subseteq \mathbb{D}^n$ implies
		\begin{equation} \label{eq:mu}
		\mu := \sup_{\mb{z} \in \mathbb{D}^n} |f(\mb{z})| < 1.
		\end{equation}
		We may assume $\mu > 0$ since otherwise $f$ is a constant function and the result holds trivially. The point $\mb{w}$ is a fixed point of $f$ if and only if it is a zero of the function $\mb{z} - f(\mb{z})$. To prove the existence of a zero we apply generalized Rouch\'e's theorem (Theorem \ref{thm:ndimRouche}) to the boundary of the polydisc $\mathbb{D}^n(\mb{0}, \rho \mb{1})$ where $\mu < \rho < 1$. On the boundary $\partial \mathbb{D}^n(\mb{0}, \rho \mb{1}) = \{ \mb{z} \in \bar{\mathbb{D}}^n(\mb{0}, \rho \mb{1}) :\; |z_i| = \rho \text { for some } i \}$, the norm of identity function $I(\mb{z}) := \mb{z}$ has larger magnitude than that of $-f(\mb{z})$ since $|I(\mb{z})| = \rho > \mu \ge |-f(\mb{z})|$, so the hypotheses of generalized Rouch\'e's theorem are satisfied. It follows that $I(\mb{z})$ and $I(\mb{z}) - f(\mb{z}) = \mb{z} - f(\mb{z})$ have the same number of zeros in $\mathbb{D}^n(\mb{0}, \rho\mb{1})$, namely one. Obviously $I(\mb{z}) - f(\mb{z})$ has no zeros in $\mathbb{D}^n \setminus \mathbb{D}^n(\mb{0}, \rho\mb{1})$, thus $f(\mb{z})$ has exactly one fixed point in $\mathbb{D}^n$.

		To be able to show convergence, we apply Schwarz's lemma. Define a new function holomorphic in $\mathbb{D}^n$ with a zero at $\mb{0}$. To this end, let $t_i : \mathbb{C}^n \to \mathbb{C}$ be a M\"obius transformation which maps $\mathbb{D}^n$ onto itself and sends $w_i$ to zero for every $i = 1, \ldots, n$. Specifically, let
		\begin{equation}
		t_i(\mb{z}) = \frac{(\mb{e}^i)^\top (\mb{z} - \mb{w})}{1 - \mb{w}^H\mb{e}^i(\mb{e}^i)^\top \mb{z}}.
		\end{equation}
		Define the function $t := (t_1, \ldots, t_n)^\top$, then it is seen that the function $g := t \circ f \circ t^{-1}$ is holomorphic in $\mathbb{D}^n$ and has fixed point $\mb{0}$. Moreover, it is bounded by a proper subset of the unit polydisc since
		\begin{equation}
		g(\mathbb{D}^n) = t \circ f \circ t^{-1}(\mathbb{D}^n) \subseteq t(\mathbb{D}^n(\mb{0}, \mu\mb{1})) \subseteq \mathbb{D}^n(\mb{0}, \kappa\mb{1})
		\end{equation}
		for some $\kappa < 1$, where the first containment is due to \eqref{eq:mu}, the second is due to the fact that
		\begin{equation} \label{eq:maxt}
		\max_{|\mb{z}| = \rho} |t(\mb{z})| = \frac{\rho + |\mb{w}|}{1 + \rho|\mb{w}|}
		\end{equation}
		and $\max_{|\mb{z}| = \rho} |t(\mb{z})|$ is increasing in $\rho$ for $0 \le \rho < 1$.
		
		The equation \eqref{eq:maxt} provides an upper bound for $\kappa$, which can be obtained by substituting $\rho$ and $|\mb{w}|$ by $\mu$:
		\begin{equation}
		\kappa \le \frac{2\mu}{1 + \mu^2}.
		\end{equation}
		We may assume $\kappa \neq 0$ since otherwise $f$ is constant and the convergence is trivial. Since $g_i(\mb{0}) = 0$ and $|g_i(\mb{s})| \le \kappa$ for $\mb{s} \in \mathbb{D}^n$, Lemma \ref{thm:ndimSchwarz} ensures that $|g(\mb{s})| \le \kappa |\mb{s}|$ for all $\mb{s} \in \mathbb{D}^n$. We may denote $\mb{s}^n := t(\mb{z}^n)$ for $n = 0, 1, 2, \ldots$, it then follows that
		\begin{equation}
		\mb{s}^n = t(\mb{z}^n) = t \circ f(\mb{z}^{n-1}) = t \circ f \circ t^{-1}(\mb{s}^{n-1}) = g(\mb{s}^{n-1}).
		\end{equation}
		Since $|g(\mb{s})| \le \kappa |\mb{s}|$ and $\kappa < 1$ for $\mb{s} \in \mathbb{D}^n$, we know $|\mb{s}^n| \le \kappa^n |\mb{s}^0| \to 0$ as $n \to \infty$ for $\mb{s}^0 \in \mathbb{D}^n$. 
		
		On the other hand, we have
		\begin{equation} \label{eq:z-w}
		z_i^n - w_i = t^{-1}(s_i^n) - t^{-1}(0) = \frac{1 - |w_i|^2}{1 + w_i^*s_i^n}s_i^n,
		\end{equation}
		for $i = 1, 2, , \ldots, n$. To get an upper bound for $|z_i^n - w_i|$, notice that $|s_i^n| < 1$ for all $n \ge 0$, $|w_i| \le \mu < 1$, so we have
		\begin{equation} \label{eq:telescope}
		\left| \frac{1 - |w_i|^2}{1 + w_i^*s_i^n} \right| < \frac{(1 - |w_i|)(1 + |w_i|)}{1 - |w_i|} = 1 + |w_i| \le 1 + \mu.
		\end{equation}
		It then follows from \eqref{eq:z-w} and \eqref{eq:telescope} that the sequence $\{ \mb{z}^n \}$ converges to the fixed point $\mb{w}$ for any $\mb{s}^0 \in \mathbb{D}^n$ since
		\begin{equation}
		|\mb{z}^n - \mb{w}| < (1 + \mu)\kappa^n|\mb{s}^0| \le (1 + \mu) \left( \frac{2\mu}{1 + \mu^2} \right)^n,
		\end{equation}
		for $n = 0, 1, 2, \ldots$.


		Now let $\mathbb{D}^n(\mb{c}, \mb{r})$ be arbitrary $n$-dimensional polydisc with radius $\mb{r}$ centered at $\mb{c}$. Let $q_i(\mb{z}) = (\mb{e}^i)^\top(\mb{z} - \mb{c})/r_i :\; \mathbb{D}^n(\mb{c}, \mb{r}) \to \mathbb{D}$ be the affine map that projects the polydisc to the $i$th coordinate and then sends the $i$th disk $\mathbb{D}(c_i, r_i)$ into the unit disk $\mathbb{D}$. Denote $q := (q_1, q_2, \ldots, q_n) : \mathbb{D}^n(\mb{c}, \mb{r}) \to \mathbb{D}^n$. The assertion that $\mb{w} \in \mathbb{D}^n(\mb{c}, \mb{r})$ is a fixed point of $f$ is equivalent to the assertion that $q(\mb{w})$ is a fixed point of $h := q \circ f \circ q^{-1}$ since if $\mb{w} = f(\mb{w})$, then
		\begin{equation} \label{eq:hgfg-1}
		h(q(\mb{w})) = q \circ f \circ q^{-1} (q(\mb{w})) = q \circ f(\mb{w}) = q(\mb{w}),
		\end{equation}
		and if $q(\mb{w})$ is a fixed point of $h$, then
		\begin{equation} \label{eq:fg-1hg}
		f(\mb{w}) = q^{-1} \circ h \circ q(\mb{w}) = q^{-1} \circ q(\mb{w}) = \mb{w}.
		\end{equation}
		
		The definition of $h$ above implies that: 1) the function $h$ is holomorphic in $\mathbb{D}^n$ and continuous on $\bar{\mathbb{D}}^n$ since $q$ is biholomorphic on $\bar{\mathbb{D}}^n$, $f$ is holomorphic in $\mathbb{D}^n(\mb{c}, \mb{r})$, continuous on $\bar{\mathbb{D}}^n(\mb{c}, \mb{r})$ and maps the closure $\bar{\mathbb{D}}^n(\mb{c}, \mb{r})$ into $\mathbb{D}^n(\mb{c}, \mb{r})$; and 2) $\bar{h}(\mathbb{D}^n)$ is contained in $\mathbb{D}^n$. The second statement follows from
		\begin{align}
		\bar{h}(\mathbb{D}^n) = \bar{q} \circ f \circ q^{-1} (\mathbb{D}^n) = \bar{q} \circ f(\mathbb{D}^n(\mb{c}, \mb{r})) \subseteq \mathbb{D}^n,
		\end{align}
		We have thus shown that $h$ satisfies the hypotheses of the theorem in the special case in which $\mathbb{D}^n(\mb{c}, \mb{r})$ is the unit polydisc and thus has exactly one fixed point. Thanks to \eqref{eq:hgfg-1} and \eqref{eq:fg-1hg}, $f$ has exactly one fixed point.

		To show the moreover statement, let $\mb{s}^n := q(\mb{z}^n)$ for $n = 0, 1, 2, \ldots$. Similar to the argument above, we know $\mb{s}^n = h(\mb{s}^{n-1})$ for $n = 1, 2, \ldots$. As we have shown, $\{ \mb{s}^n \}$ converges to $q(\mb{w})$ for any $\mb{s}^0 \in \mathbb{D}^n$ in such a manner that
		\begin{equation} \label{eq:s-q}
		|\mb{s}^n - q(\mb{w})| < (1 + \mu) \left( \frac{2\mu}{1 + \mu^2} \right)^n, \quad n = 0, 1, 2, \ldots
		\end{equation}
		where $\mu = \sup_{\mb{z} \in \mathbb{D}^n} |h(\mb{z})| < 1$. To show the convergence of $\{\mb{z}^n\}$, note that $\mb{s}^n = q(\mb{z}^n)$, so
		\begin{equation} \label{eq:s-q2}
		|\mb{s}^n - q(\mb{w})| = |\diag^{-1}(\mb{r})(\mb{z}^n - \mb{w})| \ge |(\mb{z}^n - \mb{w})| / |\mb{r}|,
		\end{equation}
		denoting $j := \arg\max |z_i^n - w_i|$ and $k := \arg\max r_i$, it is easy to see the last inequality holds since
		\begin{equation}
		\frac{|\mb{z}^n - \mb{w}|}{|\mb{r}|} = \frac{|z_j^n - w_j|}{r_k} \le \frac{|z_j^n - w_j|}{r_j} \le |\diag^{-1}(\mb{r})(\mb{z}^n - \mb{w})|.
		\end{equation}
		Integrating \eqref{eq:s-q} and \eqref{eq:s-q2}, we arrive at
		\begin{equation}
		|\mb{z}^n - \mb{w}| < |\mb{r}|(1 + \mu) \left( \frac{2\mu}{1 + \mu^2} \right)^n, \quad n = 0, 1, 2, \ldots
		\end{equation}
		for any $\mb{z}^0 \in \mathbb{D}^n(\mb{c}, \mb{r})$. The sequence $\{ \mb{z}^n \}$ converges to $\mb{w}$ since $|2\mu / (1 + \mu^2)| < 1$ for any $0 \le \mu < 1$, which completes the proof.
		%
		%
	\end{proof}

	\subsection{Application to power flow equations}
	
	We present the main result of the paper in this section. Specifically, we provide a complete characterization of the \emph{existence} and \emph{uniqueness} of power flow solution in a specific region, \emph{convergence} of fixed point iteration to the solution, as well as the `\emph{solutionless}' region in voltage space where no solution lies, all of which follow from Lemma \ref{thm:invariant}, Theorem \ref{thm:exist} and Theorem \ref{thm:ndimHenrici}.
	
	\begin{thm}[Existence, uniqueness, and convergence of fixed point power flow solution] \label{thm:mainb}
		Given load power $\mb{S}_L^0$ and power flow solution $\mb{v}_L^0$ satisfying \eqref{eq:pfv0}, let the power flow equation $\mb{u} = F(\mb{u})$ be defined as in \eqref{eq:pfu} with load power $\mb{S}_L = \mb{S}_L^0 + \mb{\sigma}_L$. Define the positive number $\Delta$ as in \eqref{eq:Delta}, $\ubar{r}, \bar{r}$ as the infimum and supremum of $r$ in \eqref{eq:rdef}, and $\mathcal{U}(r) := \{\mb{u} \in \mathbb{C}^n :\; |u_i-1|/|u_i| < r \}$.
		Suppose the load powers satisfy condition \eqref{eq:cond}, then the following statements concerning the power flow solution $\hat{\mb{u}}$ hold:
		\begin{enumerate}
			\item There exists a unique solution $\hat{\mb{u}}$ in $\bar{\mathcal{D}} (\ubar{r})$;
			\item There are no solutions in $\mathcal{U}(\bar{r}) \setminus \bar{\mathcal{D}} (\ubar{r})$;
			\item The fixed point iteration $\mb{u}^{n+1} = F(\mb{u}^n)$ converges to $\hat{\mb{u}} \in \bar{\mathcal{D}} (\ubar{r})$ for any $\mb{u}^0 \in \mathcal{U}(\bar{r})$ in such a manner that
			\begin{equation}
			|\mb{u}^n - \hat{\mb{u}}| < \bar{r} \xi(\mb{S}_L) ( 1 + \mu ) \left( \frac{2\mu}{ 1 + \mu^2} \right)^{n/2}, \quad n = 0, 1, 2, \ldots,
			\end{equation}
			for some number $0 \le \mu < 1$.
		\end{enumerate}
	\end{thm}
	
	\begin{proof} Similar to the proof of Theorem \ref{thm:exist}, we may assume throughout the proof that $\xi_i(\mb{S}_L) \ne 0$ for all load bus $i$. In addition, let $r$ be a positive number such that $r \in (\ubar{r}, \bar{r})$.
		\begin{enumerate}
			\item Theorem \ref{thm:ndimHenrici} provides sufficient condition on the uniqueness of fixed point of holomorphic functions. The problem of directly applying the theorem to show the uniqueness of fixed point of power flow equations lies in the fact that $F(\mb{u})$ is not holomorphic due to the presence of complex conjugation. However, the problem can be circumvented by defining the iterated power flow equations as the composition of $F$ with itself as $F^2(\mb{u}) = F \circ F(\mb{u})$.
			We easily see that 1) the function $F^2(\mb{u})$ is holomorphic on $\bar{\mathcal{D}}(r)$ and 2) based on Lemma \ref{thm:invariant}, $F^2(\bar{\mathcal{D}}(r)) \subseteq \mathcal{D}(r)$. To prove $F^2(\mb{u})$ has a unique fixed point in $\mathcal{D}(r)$ by Theorem \ref{thm:ndimHenrici}, we need to show $\bar{F}^2(\mathcal{D}(r)) \subseteq \mathcal{D}(r)$. It follows from point 2) above that this amounts to showing $\bar{F}^2(\mathcal{D}(r)) \subseteq F^2(\bar{\mathcal{D}}(r))$, which holds true since $\bar{F}^2(\mathcal{D}(r))$ is the intersection of all closed sets containing $F^2(\mathcal{D}(r))$ including $F^2(\bar{\mathcal{D}}(r))$ (which is closed as it is the image of a continuous function over compact set). It then follows that Theorems \ref{thm:exist} and \ref{thm:ndimHenrici} ensures the existence and uniqueness of fixed point for $\mb{u} = F^2(\mb{u})$ in $\mathcal{D}(r)$, respectively.
			
			Therefore, we know $\mb{u} = F(\mb{u})$ has at least one fixed point in $\mathcal{D}(r)$ due to Theorem \ref{thm:exist} and $\mb{u} = F^2(\mb{u})$ has exactly one fixed point in $\mathcal{D}(r)$. Since any fixed point of $\mb{u} = F(\mb{u})$ is also a fixed point of $\mb{u} = F^2(\mb{u})$, $\mb{u} = F(\mb{u})$ has exactly one fixed point in $\mathcal{D}(r)$ for $r \in (\ubar{r}, \bar{r})$. Therefore, $\mb{u} = F(\mb{u})$ admits a unique solution in $\bigcap_{r \in (\ubar{r},\bar{r})} \mathcal{D}(r) = \bar{\mathcal{D}}(\ubar{r})$.

			\item Suppose for the sake of contradiction there exists a fixed point $\hat{\mb{u}} \in \mathcal{U}(\bar{r}) \setminus \bar{\mathcal{D}}(\ubar{r})$. Since $F(\mathcal{U}(\bar{r}) \setminus \bar{\mathcal{D}}(\ubar{r})) \subseteq F(\mathcal{U}(\bar{r})) \subseteq \mathcal{D}(\bar{r})$, we know $\hat{\mb{u}}$ lies in $\mathcal{D}(\bar{r}) \setminus \bar{\mathcal{D}}(\ubar{r})$. This is impossible since we know from item (i) above that the unique solution to $F(\mb{u}) = \mb{u}$ in $\mathcal{D}(\bar{r})$ lies in $\bar{\mathcal{D}}(\ubar{r})$.

			\item Given the sequence $\{ \mb{u}^n \}$ defined by the power flow iteration $\mb{u}^{n+1} = F(\mb{u}^n)$, the subsequence comprising all odd terms of $\{ \mb{u}^n \}$ can be generated by the iteration $\mb{u}^{2k} = F^2(\mb{u}^{2k-2})$ for $k = 1, 2, \ldots$ while the subsequence comprising all even terms can be generated by the iteration $\mb{u}^{2k+1} = F^2(\mb{u}^{2k-1})$ for $k = 1, 2, \ldots$. For $\mb{u}^0 \in \mathcal{U}(\bar{r})$, we have $F(\mb{u}^0) \in \mathcal{D}(\bar{r})$ and subsequently $\mb{u}^k \in \mathcal{D}(\bar{r})$ for any $k > 0$ based on Lemma \ref{thm:invariant}. In particular, both $\mb{u}^1$ and $\mb{u}^2$ are in $\mathcal{D}(\bar{r})$. It follows from Theorem \ref{thm:ndimHenrici} that both subsequences converge to the unique fixed point in $\mathcal{D}(\bar{r})$, which means the sequence $\{ \mb{u}^n \}$ itself converges to the unique fixed point in $\mathcal{D}(\bar{r})$. Furthermore, the fixed point is in $\bar{\mathcal{D}}(\ubar{r})$ based on item (i) above.
			
			Now we show the convergence rate. Given $\mb{u}^0 \in \mathcal{U}(\bar{r})$ and denote $r := |\mathbf{I} - \diag^{-1}(\mb{u}^0)|$, there exists $\epsilon > 0$ such that $r' := \bar{r} - \epsilon > \max\{ \ubar{r}, r \}$. Define
			\begin{equation} \label{eq:c24}
			\mu := \frac{\max_{\mb{u} \in \bar{\mathcal{D}}(r')}  |\mathbf{I} - \diag^{-1}(\mb{u}^*)| }{r'},
			\end{equation}
			we know from the proof of Theorem \ref{thm:invariant} that $\bar{\mathcal{D}}(r') \subseteq \mathcal{U}(r')$ and consequently $\mu < 1$. In addition, we know
			\begin{equation} \label{eq:Fsq}
			F^2(\bar{\mathcal{D}}(r')) \subseteq F(\mathcal{D}(r')) \subseteq \left\{ \mb{u} \in \mathbb{C}^n :\; |1 - \eta_i(\mb{\sigma}_L) - u_i| < \mu \cdot r' \xi_i(\mb{S}_L) \right\},
			\end{equation}
			where the second set inclusion comes from \eqref{eq:c24}. Let $q_i(\mb{u}) = \frac{ u_i - \eta_i(\mb{\sigma}_L) }{ r'\xi_i(\mb{S}_L) } :\; \mathcal{D}(r') \to \mathbb{D}$  be the affine map that projects the polydisc $\mathcal{D}(r')$ to the $i$th dimension and then sends the $i$th disk $\mathbb{D}(\eta_i(\mb{\sigma}_L), r'\xi_i(\mb{S}_L))$ into the unit disk $\mathbb{D}$. Denote $q := (q_1, q_2, \ldots, q_n) : \mathcal{D}(r') \to \mathbb{D}^n$.
			If we define $h(\mb{z}) := q \circ F^2 \circ q^{-1}(\mb{z})$, it follows from \eqref{eq:Fsq} that $\mu$ is an upper bound of $\| h(\mb{z}) \|_\infty$ for $\mb{z} \in \mathbb{D}^n$ since
			\begin{equation}
			\sup_{\mb{z} \in \mathbb{D}^n} \| h(\mb{z}) \|_\infty = \sup_{\mb{z} \in \mathcal{D}(r')} \| q \circ F^2(\mb{z}) \|_\infty < \mu.
			\end{equation}
			We then know from the proof of Theorem \ref{thm:ndimHenrici} that for the sequence $\{ \mb{u}^{2n} \}$ generated by $\mb{u}^{2n+2} = F^2(\mb{u}^{2n}), n = 0, 1, 2, \ldots$, we have
			\begin{equation}
			|\mb{u}^{2n} - \hat{\mb{u}}| < r'\xi(\mb{S}_L) (1 + \mu) \left( \frac{2\mu}{ 1 + \mu^2 } \right)^n, n = 0, 1, 2, \ldots
			\end{equation}
			for any $\mb{u}^0 \in \mathcal{U}(r')$. Given $\mb{u}^1 = F(\mb{u}^0) \in \mathcal{D}(r')$, we can verify that 
			\begin{equation}
			|\mb{u}^1 - \hat{\mb{u}}| < r'\xi(\mb{S}_L)(1 + \mu) \sqrt{ \frac{2\mu}{1 + \mu^2} }.
			\end{equation}
			In addition, the sequence $\{ \mb{u}^{2n+1} \}$ by $\mb{u}^{2n+1} = F^2(\mb{u}^{2n-1})$, $n = 1, 2, 3, \ldots$ has the same convergence rate. It follows that the sequence $\{ \mb{u}^n \}$ generated by $\mb{u}^{n+1} = F(\mb{u}^n)$ has the following convergence rate:
			\begin{equation} \label{eq:c29}
			|\mb{u}^{n} - \hat{\mb{u}}| < r'\xi(\mb{S}_L)(1 + \mu) \left( \frac{2\mu}{ 1 + \mu^2 } \right)^{n/2}
			\end{equation}
			for the given $\mb{u}^0$. In fact, for any $\mb{u}^0 \in \mathcal{U}(\bar{r})$, there exists such $0 \le \mu < 1$ and the $r'$ factor in \eqref{eq:c29} is upper bounded by $\bar{r}$. \qedhere
		\end{enumerate}
	\end{proof}

	\section{Relationship to Existing Conditions} \label{sect:relation}
	
	In this section we compare the proposed condition \eqref{eq:cond} with two sharpest results in the literature known so far: conditions in \cite{WangC16} and \cite{Dvijotham18}. The two conditions are incomparable, while it was empirically shown that the certified solvability set by condition \cite{WangC16} is generally `smaller' than the one by \cite{Dvijotham18}. We briefly introduce the two existing conditions in Theorems \ref{thm:WBBP} and \ref{thm:DNT} below. To be consistent with the adopted model and notations in Section \ref{sect2}, the two results are slightly rephrased and generalized without proof. We then give proof of dominance of the proposed condition \eqref{eq:cond} over the two existing ones in Proposition \ref{thm:compare} by showing the certified solvability set by the proposed condition \eqref{eq:cond} contain those given by the two existing conditions.
	
	\begin{thm}[\hspace{1sp}{\cite[Thm. 1]{WangC16}}] \label{thm:WBBP}
		Given load power $\mb{S}_L^0$ and power flow solution $\mb{v}_L^0$ satisfying \eqref{eq:pfv0} where $\xi(\mb{S}_L^0) < 1$, the power flow equation $\mb{u} = F(\mb{u})$ defined in \eqref{eq:pfu} with load power $\mb{S}_L = \mb{S}_L^0 + \mb{\sigma}_L$ admits a unique solution in $\{ \mb{u} \in \mathbb{C}^n \mid   1-r \le |u_i| \le 1+r \}$ where
		\begin{equation}
		r = \frac{1 - \xi(\mb{S}_L^0) - \sqrt{\left( 1 - \xi(\mb{S}_L^0) \right)^2 - 4\xi(\mb{\sigma}_L)} }{2}
		\end{equation}
		when
		\begin{equation} \label{eq:wang_condition}
		\left( 1 - \xi(\mb{S}_L^0) \right)^2 - 4\xi(\mb{\sigma}_L) > 0.
		\end{equation}
	\end{thm}
	
	\begin{thm}[\hspace{1sp}{\cite[Sect. IV-A]{Dvijotham18}}] \label{thm:DNT}
		Given load power $\mb{S}_L^0$ and power flow solution $\mb{v}_L^0$ satisfying \eqref{eq:pfv0}, the power flow equation $\mb{u} = F(\mb{u})$ defined in \eqref{eq:pfu} with load power $\mb{S}_L = \mb{S}_L^0 + \mb{\sigma}_L$ admits at least one solution in $\{ \mb{u} \in \mathbb{C}^n \mid   1/(1+r) \le |u_i| \le 1/(1-r) \}$ where
		\begin{equation}
		r = \frac{1 - \xi(\mb{S}_L) - \eta(\mb{\sigma}_L) - \sqrt{( 1 - \xi(\mb{S}_L) - \eta(\mb{\sigma}_L) )^2 - 4\xi(\mb{S}_L)\eta(\mb{\sigma}_L)}}{2\xi(\mb{S}_L)}
		\end{equation}
		when
		\begin{equation} \label{eq:dj_condition}
		\sqrt{\xi(\mb{S}_L)} + \sqrt{\eta(\mb{\sigma}_L)} \le 1.
		\end{equation}
	\end{thm}
	
	
	
	
	\subsection{Theoretical justification}
	
	Given load power $\mb{S}_L^0$ and power flow solution $\mb{v}_L^0$ satisfying \eqref{eq:pfv0}, let the solvability sets $\mathcal{S}_p$, $\mathcal{S}_w$, and $\mathcal{S}_d$ be the sets of incremental load power $\mb{\sigma}_L$ satisfying \eqref{eq:cond}, \eqref{eq:wang_condition}, and \eqref{eq:dj_condition}, respectively. Note that both Theorems \ref{thm:mainb} and \ref{thm:WBBP} provide uniqueness guarantees while Theorem \ref{thm:DNT} only guarantees solution existence. This causes $\mathcal{S}_p$ and $\mathcal{S}_w$ to be open while $\mathcal{S}_d$ is closed. Therefore, when comparing the strength of conditions \eqref{eq:cond} and \eqref{eq:dj_condition}, we compare $\mathcal{S}_d$ with $\bar{\mathcal{S}}_p$ --- the closure of $\mathcal{S}_p$. It should be noted that this treatment is a mere technicality and have negligible consequence in practice.
	
	\begin{pro} \label{thm:compare}
		Given load power $\mb{S}_L^0$ and power flow solution $\mb{v}_L^0$ satisfying \eqref{eq:pfv0}, let the solvability sets $\mathcal{S}_p$, $\mathcal{S}_w$, and $\mathcal{S}_d$ be the sets of incremental power injection $\mb{\sigma}_L$ satisfying \eqref{eq:cond}, \eqref{eq:wang_condition}, and \eqref{eq:dj_condition}, respectively, then $\mathcal{S}_w \subseteq \mathcal{S}_p$ and $\mathcal{S}_d \subseteq \bar{\mathcal{S}}_p$ hold.
		
		Moreover, the proposed condition strictly dominates conditions \eqref{eq:wang_condition} and \eqref{eq:dj_condition}, or $\mathcal{S}_w \subsetneq \mathcal{S}_p$ and $\mathcal{S}_d \subsetneq \bar{\mathcal{S}}_p$, when $\{ \mb{0} \} \subsetneq \mathcal{S}_p$.
	\end{pro}
	\begin{proof}
		
		To show the proposed condition dominates \eqref{eq:wang_condition}, or $\mathcal{S}_w \subseteq \mathcal{S}_p$, we show any $\mb{\sigma}_L^w$ contained in $\mathcal{S}_w$ is in $\mathcal{S}_p$ as well. Let $\mb{S}_L^w := \mb{S}_L^0 + \mb{\sigma}_L^w$. We know  $\mb{\sigma}_L^w$ satisfies \eqref{eq:condb} since
		\begin{equation} \label{eq:d3}
		\xi(\mb{S}_L^w) - \eta(\mb{\sigma}_L^w) \le \xi(\mb{S}_L^w) + \eta(\mb{\sigma}_L^w) \le \xi(\mb{S}_L^0) + 2\xi(\mb{\sigma}_L^w) \le  \xi(\mb{S}_L^0) + 2\sqrt{\xi(\mb{\sigma}_L^w)} < 1,
		\end{equation}
		where the third inequality comes from the observation that \eqref{eq:wang_condition} implies $\xi(\mb{\sigma}_L^w) < 1$, and the last inequality is obtained by moving the second term in \eqref{eq:wang_condition} to the right, take square root on both sides, and rearrange terms. To show $\mb{\sigma}_L^w$ satisfies \eqref{eq:cond_bl}, we can instead show the left hand side of \eqref{eq:cond_bl} is less than or equal to $4\xi(\mb{\sigma}_L^w) - \xi(\mb{S}_L^0)^2 + 2\xi(\mb{S}_L^0)$, which is less than 1 by \eqref{eq:wang_condition}. For $\mb{\sigma}_L^w$ satisfying \eqref{eq:wang_condition} and any $i \in \mathcal{N}_L$, we have 
		\begin{subequations} \label{eq:d4}
			\begin{align}
			\gamma_i(\mb{S}_L^w, \mb{\sigma}_L^w) &= 2(\xi_i(\mb{S}_L^w) + \re(\eta_i(\mb{\sigma}_L^w))) - \xi_i(\mb{S}_L^w)^2 - |\eta_i(\mb{\sigma}_L^w)|^2 \\
			&\le 2(\xi_i(\mb{S}_L^w) + |\eta_i(\mb{\sigma}_L^w)|) - \xi_i(\mb{S}_L^w)^2 - |\eta_i(\mb{\sigma}_L^w)|^2 \label{eq:d4b} \\
			&\le 2(\xi(\mb{S}_L^w) + \xi(\mb{\sigma}_L^w)) - \xi(\mb{S}_L^w)^2 - \xi(\mb{\sigma}_L^w)^2, \label{eq:d4c}
			\end{align}
		\end{subequations}
		where we replace $\xi_i(\mb{S}_L^w)$ and $|\eta_i(\mb{\sigma}_L^w)|$ by $\xi(\mb{S}_L^w)$ and $\xi(\mb{\sigma}_L^w)$ in the second inequality since $-x^2 + 2x$ is increasing for $x < 1$ and $\xi(\mb{S}_L^w), \xi(\mb{\sigma}_L^w) < 1$ based on \eqref{eq:d3}. Subtract $4\xi(\mb{\sigma}_L^w) - \xi(\mb{S}_L^0)^2 + 2\xi(\mb{S}_L^0)$ from the left hand side of \eqref{eq:cond_bl}, denote the difference by $\delta$, and apply \eqref{eq:d4c}, we have
		\begin{equation}
		\delta \le \left( \xi(\mb{S}_L^w) - \xi(\mb{S}_L^0) - \xi(\mb{\sigma}_L^w) \right)\left( 2 + \xi(\mb{\sigma}_L^w) - \xi(\mb{S}_L^w) - \xi(\mb{S}_L^0) \right),
		\end{equation}
		which is nonpositive since the first term is nonpositive and the second term is lower bounded by $2 - 2\xi(\mb{S}_L^0) > 0$. This shows $\mathcal{S}_w \subseteq \mathcal{S}_p$.
		
		We now show $\mathcal{S}_d \subseteq \bar{\mathcal{S}}_p$. It is clear from \eqref{eq:dj_condition} that for any $\mb{\sigma}_L^d \in \mathcal{S}_d$, \eqref{eq:condb} is satisfied. We are left to show $\mb{\sigma}_L^d$ satisfies \eqref{eq:cond_bl}. Let $\mb{S}_L^d := \mb{S}_L^0 + \mb{\sigma}_L^d$. Similar to \eqref{eq:d4}, we have
		\begin{equation} \label{eq:d6}
		\gamma(\mb{S}_L^d, \mb{\sigma}_L^d) + 2\xi(\mb{S}_L^d) \eta(\mb{\sigma}_L)^d \le 
		2(\xi(\mb{S}_L^d) + \eta(\mb{\sigma}_L^d)) - (\xi(\mb{S}_L^d) - \eta(\mb{\sigma}_L^d))^2,
		\end{equation}
		so we only need to show the right hand side of \eqref{eq:d6} is less than or equal to 1. The following two conditions hold for some nonnegative number $p$ by respectively squaring once and twice on both sides of \eqref{eq:dj_condition} and rearrange terms:
		\begin{subequations}
			\begin{align}
			\xi(\mb{S}_L^d) + \eta(\mb{\sigma}_L^d) &= 1 - p - 2 \sqrt{\xi(\mb{S}_L^d) \eta(\mb{\sigma}_L^d)}, \\
			\left(\xi(\mb{S}_L^d) - \eta(\mb{\sigma}_L^d)\right)^2 &= (1-p)^2 - 4(1-p)\sqrt{\xi(\mb{S}_L^d)\eta(\mb{\sigma}_L^d)}.
			\end{align}
		\end{subequations}
		Substitute the two relations into the right hand side of \eqref{eq:d6} confirms it is indeed less than or equal to 1:
		\begin{equation}
		1 - p^2 - 4p\sqrt{\xi(\mb{S}_L^d)\eta(\mb{\sigma}_L^d)} \le 1.
		\end{equation}

		
		
		
		In summary, we have shown above that $\mathcal{S}_w \subseteq \mathcal{S}_p$ and $\mathcal{S}_d \subseteq \bar{\mathcal{S}}_p$. Now we show $\mathcal{S}_w \subsetneq \mathcal{S}_p$ when $\{\mb{0}\} \subsetneq \mathcal{S}_p$. It is easy to see that there exists $\mb{\sigma}_L^e$ such that \eqref{eq:wang_condition} holds with equality. Let $k := \argmax |\eta_k(\mb{\sigma}_L^e)|$, we may assume $\eta_k(\mb{\sigma}_L^e)$ has nonzero imaginary part (otherwise we can multiply $\mb{\sigma}_L^e$ by some complex number with unity modulus without changing equality of \eqref{eq:wang_condition}), then the inequality \eqref{eq:d4b} is strict and $\mb{\sigma}_L^e \in \mathcal{S}_p$. So we have identified $\mb{\sigma}_L^e \in \mathcal{S}_p \setminus \mathcal{S}_{w}$. The proof that $\mathcal{S}_d \subsetneq \bar{\mathcal{S}}_p$ is similar: we can find $\mb{\sigma}_L^d$ on the boundary of $\mathcal{S}_d$ such that $\eta_k(\mb{\sigma}_L^d)$ has nonzero imaginary part where $k := \argmax |\eta_k(\mb{\sigma}_L^d)|$. It follows from the telescoping of \eqref{eq:d6} that $\mathcal{S}_d \in \mathcal{S}_p$. By continuity, there is a $\mb{\sigma}_L$ in the neighborhood of $\mb{\sigma}_L^d$ contained in $\mathcal{S}_p \setminus \mathcal{S}_d$. This completes the proof.
	\end{proof}
	
	\subsection{Computational results}
	
	We perform computational experiment to numerically compare the maximum load powers certified by the three conditions. Ten standard IEEE test systems are used for the experiment, the data of which are available in \textsc{Matpower} package \cite{Zimmerman11}. We assume the power flow solution $\mb{v}_L^0$ to the base loading $\mb{S}_L^0$ provided in the data sets are known and use it to construct the normalized impedance matrix $\tilde{\mb{Z}}$. We are interested in certifying the maximum scaling factor $\lambda$ such that the power flow is still guaranteed to be solvable with loading $(1+\lambda)\mb{S}^0_L$ by the three conditions. The computational results are shown in Tables \ref{tb:compa} and \ref{tb:compb}. It is seen that the proposed condition consistently outperforms the other two, which serves as partial numerical evidence of the dominance of the proposed condition.
	
	\begin{table}[ht]
		\centering
		\caption{Lower bounds of solvability limits with base loading obtained using the proposed condition and two existing conditions versus the true solvability limits, the table shows the maximum scaling factors $1+\lambda$ certified by each condition.}
		\begin{tabular}{lccccc}
			Test case       & Proposed & \cite{Dvijotham18} & \cite{WangC16} & Actual value \\ \midrule
			9-bus system & $2.4676$ & $2.0493$ & $2.0399$ & $2.6577$ \\ 
			14-bus system & $4.3862$ & $3.7605$ & $3.6144$ & $5.3320$ \\ 
			24-bus system & $2.4101$ & $1.9656$ & $1.9091$ & $2.7928$ \\ 
			30-bus system & $5.4665$ & $4.9966$ & $4.9346$ & $6.0160$ \\ 
			39-bus system & $2.1826$ & $1.7650$ & $1.6846$ & $2.4730$ \\ 
			57-bus system & $1.4719$ & $1.3764$ & $1.3454$ & $1.9074$ \\ 
			118-bus system & $4.7987$ & $4.1189$ & $3.8447$ & $5.4479$ \\ 
			300-bus system & $1.0558$ & $1.0284$ & $1.0047$ & $1.6585$ \\ 
			1354-bus system & $1.3595$ & $1.2012$ & $1.1597$ & $1.5332$ \\ 
			2383-bus system & $1.5708$ & $1.3955$ & $1.3683$ & $1.9739$ \\  \bottomrule
		\end{tabular}
		\label{tb:compa}
	\end{table}
	
	\begin{table}[ht]
		\centering
		\caption{Relative errors of solvability limit approximations with base loading $\mb{S}_L^0$ obtained using the proposed condition and two existing conditions.}
		\begin{tabular}{lccc}
			Test case       & Proposed & \cite{Dvijotham18} & \cite{WangC16} \\ \midrule
			9-bus system & $7.15\%$ & $22.89\%$ & $23.24\%$ \\ 
			14-bus system & $17.74\%$ & $29.47\%$ & $32.21\%$ \\ 
			24-bus system & $13.70\%$ & $29.62\%$ & $31.64\%$ \\ 
			30-bus system & $9.13\%$ & $16.95\%$ & $17.98\%$ \\ 
			39-bus system & $11.74\%$ & $28.63\%$ & $31.88\%$ \\ 
			57-bus system & $22.83\%$ & $27.84\%$ & $29.46\%$ \\ 
			118-bus system & $11.92\%$ & $24.40\%$ & $29.43\%$ \\ 
			300-bus system & $36.34\%$ & $37.99\%$ & $39.42\%$ \\ 
			1354-bus system & $11.33\%$ & $21.65\%$ & $24.36\%$ \\ 
			2383-bus system & $20.42\%$ & $29.30\%$ & $30.68\%$ \\ 
			\textbf{Average} & $\mb{16.23\%}$ & $\mb{26.87\%}$ & $\mb{29.03\%}$ \\ \bottomrule
		\end{tabular}
		\label{tb:compb}
	\end{table}

\end{document}